\documentclass[preprint]{imsart}

\RequirePackage{amsthm,amsmath,amsfonts,amssymb}
\RequirePackage[numbers]{natbib}
\RequirePackage[colorlinks,citecolor=blue,urlcolor=blue]{hyperref}
\RequirePackage{graphicx}
\RequirePackage{xcolor}
\RequirePackage{commath}

\startlocaldefs
\newtheorem{theorem}{Theorem}[section]
\newtheorem{proposition}[theorem]{Proposition}
\newtheorem{corollary}[theorem]{Corollary}

\theoremstyle{remark}
\newtheorem{definition}[theorem]{Definition}

\ifx\example\undefined
\newtheorem{example}[theorem]{Example}

\fi
\newcommand{\R}{{\mathbb R}}
\newcommand{\pfc}{{\mathcal F_c}}
\newcommand{\pkc}{{\mathcal K_c}}
\newcommand{\N}{{\mathbb N}}

\endlocaldefs

\begin{document}
	
	\begin{frontmatter}
		\title{Simplicial depths for fuzzy random variables}
		%\title{A sample article title with some additional note\thanksref{t1}}
		\runtitle{Simplicial Fuzzy Depth}
		%\thankstext{T1}{A sample additional note to the title.}
		
		\begin{aug}
			%%%%%%%%%%%%%%%%%%%%%%%%%%%%%%%%%%%%%%%%%%%%%%
			%%Only one address is permitted per author. %%
			%%Only division, organization and e-mail is %%
			%%included in the address.                  %%
			%%Additional information can be included in %%
			%%the Acknowledgments section if necessary. %%
			%%%%%%%%%%%%%%%%%%%%%%%%%%%%%%%%%%%%%%%%%%%%%%
			\author[A]{\fnms{LUIS} \snm{GONZ\'ALEZ-DE LA FUENTE}\ead[label=e1]{gdelafuentel@unican.es}},
			\author[A]{\fnms{ALICIA} \snm{NIETO-REYES}\ead[label=e2]{alicia.nieto@unican.es}}
			\and
			\author[B]{\fnms{PEDRO} \snm{TER\'AN}\ead[label=e3]{teranpedro@uniovi.es}}
			%%%%%%%%%%%%%%%%%%%%%%%%%%%%%%%%%%%%%%%%%%%%%%
			%% Addresses                                %%
			%%%%%%%%%%%%%%%%%%%%%%%%%%%%%%%%%%%%%%%%%%%%%%
			\address[A]{Departamento de Matem\'aticas, Estad\'istica y Computaci\'on,
				Universidad de Cantabria (Spain),
				%\printead{e1,e2}
			}
			
			\address[B]{Departamento de Estad\'istica e Investigaci\'on Operativa y Did\'actica de las Matem\'aticas,
				Universidad de Oviedo (Spain),
				%\printead{e3}
			}
		\end{aug}
		
		\begin{abstract}
			The recently defined concept of a statistical depth function for fuzzy sets provides a theoretical framework for ordering fuzzy sets with respect to the distribution of a fuzzy random
			variable. One of the most used and studied statistical depth function for multivariate data is simplicial depth, based on multivariate simplices. We introduce a notion of pseudosimplices generated by fuzzy sets and propose
			three plausible generalizations of simplicial depth to fuzzy sets.  Their theoretical properties are analyzed and the behavior of the proposals illustrated through a study of both synthetic and real data.			\end{abstract}
		
		\if0
		...........
		..........
		\fi
		
		
	\end{frontmatter}
	%%%%%%%%%%%%%%%%%%%%%%%%%%%%%%%%%%%%%%%%%%%%%%
	%% Please use \tableofcontents for articles %%
	%% with 50 pages and more                   %%
	%%%%%%%%%%%%%%%%%%%%%%%%%%%%%%%%%%%%%%%%%%%%%%
	%\tableofcontents
	
	\section{Introduction}
		In the general framework  of fuzzy data, the data consists of classes of objects with a continuum of grades of membership \citep{zadehfuzzysets}. They are generally represented as functions from $\mathbb{R}^p$ to $[0,1],$ as
	opposed to multivariate data which are points in $\mathbb{R}^p.$ On the other hand, statistical depth functions are a quantification of
	the intuitive notion that the median is the point that is most `in the middle'. They do so
	by providing a center-outward ordering of the points in a space with respect to a probability distribution or data set. While this task is trivial in the real line, in the sense that moving outward is just going towards
	$-\infty$ or $\infty$, it becomes harder for multivariate data (and even more so for more complex types of data) as no natural total order is present.
	
	To understand some of the challenges involved, consider the first idea one might have, which is to apply the median, coordinate-wise, to obtain a multivariate median in $\R^p$. 	The
	coordinate-wise median may lie outside the convex hull of the data, against the idea that the median should be as much `in the middle' of the data as possible. Moreover, by changing the coordinate system (which does not
	affect the data themselves, only how we represent them) the coordinate-wise median of the data set can be modified. Even in simple cases, like the vertices of an equilateral triangle and its center of mass, it fails to provide the intuitive solution that the innermost point is the center of mass \citep{Rafa}.

	The notion of a statistical depth function (giving each point in a space a depth value with respect to a sample or a distribution on the space, as a measure of its centrality) opens an avenue for extending rank-based and
	quantile-based statistical procedures from the real line to more complex spaces. Tukey \citep{tukey} first introduced depth for multivariate data. Some pre-existent notions in multivariate analysis can be expressed in the
	language of depth. For instance, Mahalanobis distance gives rise to Mahalanobis depth; other examples are convex hull peeling depth \citep{Barnet1976} and Oja depth \citep{Oja}. Liu \citep{LiuSimplicial} introduced simplicial depth, which is one of the best known and most popular depth functions. She proved a number of nice properties which then inspired Zuo and Serfling's abstract definition of  statistical depth function,  constituted by a list of desirable properties \citep{ZuoSerfling}. In intuitive terms, these are as follows.
	\begin{itemize}
		\item[(M1)]{\em Affine invariance.} A change of coordinates should not affect the depth values.
		\item[(M2)]{\em Maximality at the center of symmetry.} If a distribution is symmetric, the deepest point should be the center of symmetry.
		\item[(M3)]{\em Monotonicity from the deepest point.} Depth values should decrease along any ray that departs from a deepest point.
		\item[(M4)]{\em Vanishing at infinity.} The depth value of $x$ should go to $0$ as its norm goes to infinity. 
	\end{itemize}
	It should be underlined that these are not clear-cut axioms. Failing to satisfy some property or other, or doing so only under some conditions, is not considered enough for a function to be excluded from being a depth
	function.
	
	Today, the number of depth functions runs in the dozens and this is a broad and active topic in non-parametric statistics. With the rise of functional data analysis and the apparition of several adaptations of multivariate
	depth notions to the functional setting, Nieto-Reyes and Battey proposed a list of desirable properties for depth functions in function (metric) spaces \citep{NietoBattey},  and an instance of depth satisfying all those
	properties in \citep{jmva}. This instance was later applied to a real data analysis in \citep{math1}. The connections between depth functions and fuzzy sets were noted by Ter\'an \citep{SMPS10,IJAR}, who showed that some
	depth functions can be rigorously interpreted as fuzzy sets and {\it vice versa}. In \citep{primerarticulo} we proposed two definitions of statistical depth for fuzzy data; although fuzzy sets are functions, these definitions list desirable properties tailored to fuzzy sets. We also generalized Tukey depth as a first example of depth for fuzzy data, and studied its properties. Sinova \citep{Sinova} also considered depth for fuzzy
	data and defined depth-trimmed means.

	It is important to show that more of the most relevant examples of depth can be adapted to the fuzzy setting. Firstly, to justify the viability of the notions of depth for fuzzy data. Secondly, to create a library of depth
	functions with guaranteed good theoretical properties in order to have them applied in practice. And thirdly, to test the abstract definitions in \citep{primerarticulo} and understand whether they are fine as they stand or
	might need to be adjusted.
	
	In this paper, we study the problem of adapting Liu's simplicial depth to the fuzzy setting. As mentioned above, it is one of the best known and most used depth functions for multivariate data. For instance, Liu et al.
	\citep{Parelius} developed  techniques to study multivariate distributional characteristics using simplicial depth, and other depth functions. The multivariate definition of simplicial depth assigns to each point $x\in\R^p$
	a depth value being the probability that $x$ lies in the convex hull of $p+1$ independent observations. Provided the distribution is continuous, with probability $1$ those observations define a $p$-dimensional simplex
	(a triangle in $\R^2$, a tetrahedron in $\R^3$, and so on) with non-empty interior, which may contain $x$ or not. If $x$ is very outlying in the distribution, the probability that the simplex will contain $x$ is very small.
	Thus $x$ is deep insofar as, loosely speaking, it is likely that the data points in a small sample `capture' $x$ among them.
	
	When extending this notion to functional data, L\'opez-Pintado and Romo \citep{LopezRomoBand} already realized that using the convex hull to determine which functions are `among' other functions is naive. We face similar
	problems in the fuzzy case. In the end, the convex hull of finitely many points is a finite-dimensional set, so in any infinite-dimensional space the vast majority of the elements in the space will be excluded from it. This
	creates a propensity to assign zero depth which will require an adaptation in line with that in \citep{LopezRomoBand}.
	
	Another obstacle is that some multivariate definitions do not transfer immediately to the fuzzy setting. For instance, Tukey depth is based on the notion of a halfspace but spaces of fuzzy sets, not being linear spaces,
	cannot be `halved' by hyperplanes so a workaround needed to be devised in \cite{primerarticulo}. In this case, simplicial depth rests on the notion of a simplex in $\R^p$, which, as will be discussed, also needs a
	workaround. That results in a plurality of ways to extend simplicial depth.

	The paper is organized as follows. Section \ref{prelim} contains the notation and background on fuzzy sets and statistical depth required for a comprehensive understanding of the next sections. An operative adaptation of
	simplices to spaces of sets and fuzzy sets is presented in Section \ref{pseudo}. The definitions of the proposed variants of simplicial depth are in Section \ref{simplicial}. Their status with respect to the desirable
	properties in the definitions of depth for fuzzy data \citep{primerarticulo} is studied in Section \ref{properties}, assuming that the distribution is `continuous' in a certain sense. Examples with real and simulated data
	are worked out in Section \ref{datasimulation}, while a discussion is presented in Section \ref{discussion}. All proofs are deferred to Section \ref{proofs}.
	
	\section{Notation and preliminaries}
	\label{prelim}
	
	The following notation is used throughout. A function  $A: \mathbb{R}^{p}\rightarrow [0,1]$ is a {\em fuzzy set} on $\mathbb{R}^{p}$ (or a fuzzy subset of $\mathbb{R}^{p}$). Let $\mathcal{F}_{c}(\mathbb{R}^{p})$ denote the
	class of all fuzzy sets $A$ on $\mathbb{R}^{p}$ such that the $\alpha$-level of $A,$ given by
	$$A_{\alpha}= \{x\in\mathbb{R}^{p}: A(x)\geq\alpha \}$$
	if $\alpha\in(0,1]$ and the closed support of $A$ if $\alpha=0$,
	is compact and convex for every $\alpha\in [0,1]$.  We will freely write `fuzzy set' to mean an element of $\mathcal{F}_{c}(\mathbb{R}^{p})$.
	
	Let $\mathcal{K}_{c}(\mathbb{R}^{p})$ be the class of non-empty  compact and convex subsets of $\mathbb{R}^{p}$. Any set $K\in\mathcal{K}_{c}(\mathbb{R}^{p})$ can be identified with a fuzzy set, its {\em indicator function}
	$\text{I}_{K} : \mathbb{R}^{p}\rightarrow\mathbb{R}$ where $\text{I}_{K}(x) = 1$ if $x\in K$ and $\text{I}_{K}(x) = 0$ otherwise.

	The unit sphere of $\mathbb{R}^{p}$ is $\mathbb{S}^{p-1} = \{x\in\mathbb{R}^{p}: \|x\|\leq 1 \}$, with $\|.\|$ denoting the Euclidean norm on $\mathbb{R}^{p}$. The symbol $=^{\mathcal{L}}$ denotes equality in distribution of
	random variables and $\mathcal{M}_{p\times p}(\mathbb{R})$ is the set of all $p\times p$ real matrices.
	
	The {\em support function} of $A\in\mathcal{F}_{c}(\mathbb{R}^{p})$ is the mapping $s_{A}: \mathbb{S}^{p-1}\times [0,1]\rightarrow\mathbb{R}$ such that $s_{A}(u,\alpha) := \sup_{v\in A_{\alpha}}\langle u,v\rangle, $
	for every $u\in\mathbb{S}^{p-1}$ and $\alpha\in [0,1]$, where $\langle \cdot ,\cdot\rangle$ denotes the usual inner product in $\mathbb{R}^{p}$.
	By \cite[Proposition 7.2]{primerarticulo}, 	\begin{equation}\label{pa} 		s_{M\cdot A}(u,\alpha) = \|M^{T}\cdot u\|\cdot s_{A}\left(\cfrac{1}{\|M^{T}\cdot u\|}\cdot M^{T}\cdot u,\alpha\right) 	\end{equation} for any
	$A\in\mathcal{F}_{c}(\mathbb{R}^{p}),$  $M\in\mathcal{M}_{p\times p}(\mathbb{R})$ being non-singular, $u\in\mathbb{S}^{p-1}$ and $\alpha\in [0,1].$
	
	In $\mathcal{F}_{c}(\mathbb{R}),$  the subclass of \textit{trapezoidal fuzzy sets} \cite[Section 10.7]{Klir} is used very often. Four values $a, b,c,d \in\mathbb{R}$ with $a\le b\leq c \le d$ determine the trapezoidal fuzzy
	set
	
	\begin{equation*}\label{trapezoidal}
		\mbox{Tra}(a,b,c,d)(x) := \left\{ \begin{array}{lcr}
			\cfrac{x - a}{b-a},  & \text{ if } &a< x<b, \\
			\\1,& \text{ if } &b\leq x\leq c, \\
			\\ \cfrac{x - d}{c-d},  & \mbox{ } \text{ if } & c< x< d, \\
			\\ 0, && \text{otherwise.}
		\end{array}
		\right.
	\end{equation*}
	
	\subsection{Arithmetics and Zadeh's extension principle}\label{Aarith}
	Let $A,B\in\mathcal{F}_{c}(\mathbb{R}^{p})$ and $\gamma\in\mathbb{R}$. The formulae 		\begin{equation}\nonumber 			(A + B)(t) := \sup_{x,y\in\mathbb{R}^{p}: x + y = t} \min\{A(x), B(y) \}, \text{ and} 		
	\end{equation} 		\begin{equation}\nonumber 			(\gamma\cdot A)(t) := \sup_{x\in\mathbb{R}^{p} : t = \gamma\cdot x} A(y) = \left\{ 			\begin{array}{lrr} 				A\left(\frac{t}{\gamma}\right),     &
			\mbox{ } \mbox{ if } & \gamma\neq 0 \\ 				\\ 				I_{\{0\}}(t), & \mbox{ } \mbox{ if } & \gamma = 0 			\end{array} 			\right. 		\end{equation} valid for arbitrary $t\in\R^p$, define an
	addition and a product by scalars in $\pfc(\R^p)$.
	
	Given $A,B\in\mathcal{F}_{c}(\mathbb{R}^{p})$, $\gamma\in [0,\infty)$, $u\in\mathbb{S}^{p-1}$ and $\alpha\in [0,1],$  a useful relation that makes use of these operations is
	\begin{equation}\label{soportesuma}
		s_{A+\gamma\cdot B}(u,\alpha) = s_{A}(u,\alpha) + \gamma\cdot s_{B}(u,\alpha).
	\end{equation}
	
	Zadeh's extension principle  \citep{zadehextension} allows a continuous, crisp, function  $f: \mathbb{R}^{p}\rightarrow\mathbb{R}^{p}$ to act on a fuzzy set $A\in\mathcal{F}_{c}(\mathbb{R}^{p}),$ obtaining
	$f(A)\in\mathcal{F}_{c}(\mathbb{R}^{p})$  with $f(A)(t) := \sup\{A(y) : y\in\mathbb{R}^{p}, f(y) = t \}$ for all $t\in\mathbb{R}^{p}$.

	\subsection{Metrics in the fuzzy setting}
	\label{MFS}
	
	We will make use of different metrics in $\mathcal{F}_{c}(\mathbb{R}^{p})$. For any fuzzy sets $A,B\in\mathcal{F}_{c}(\mathbb{R}^{p})$, let
	\begin{equation}\nonumber
		d_{r}(A,B) := \left\{ \begin{array}{lrr} 	\left(\int_{[0,1]} \left( d_{\mathcal{H}} (A_{\alpha},B_{\alpha}) \right)^{r} \dif\nu(\alpha)\right)^{1/r} & \mbox{ } \mbox{  if } & r\in [1,\infty),\\ 	\\ \sup_{\alpha\in
				[0,1]} d_{\mathcal{H}}  (A_{\alpha},B_{\alpha}) &  \mbox{ } \mbox{  if } & r = \infty,
		\end{array}
		\right.
	\end{equation}
	where $$d_{\mathcal{H}} (S,T) := \max\left\{\sup_{s\in S}\inf_{t\in T} \parallel s-t\parallel, \sup_{t\in T}\inf_{s\in S}\parallel s-t\parallel\right\}$$ defines the {\em Hausdorff metric} and $\nu$ denotes the Lebesgue
	measure in $[0,1]$. While $(\mathcal{F}_{c}(\mathbb{R}^{p}), d_{r})$ is a non-complete and separable metric space for any $r\in[1,\infty)$, the metric space
	$(\mathcal{F}_{c}(\mathbb{R}^{p}), d_{\infty})$ is non-separable and complete  \citep{diamondkloden}.
	According to \citep{diamondkloden}, it is also possible to consider $L^{r}$-type metrics for any
	$A,B\in\mathcal{F}_{c}(\mathbb{R}^{p}),$
	\begin{equation}\nonumber
		\rho_{r}(A,B) := \left(\int_{\mathbb{S}^{p-1}}\int_{[0,1]}|s_{A}(u,\alpha)-s_{B}(u,\alpha)|^{r}  \dif\nu(\alpha) \dif\mathcal{V}_{p}(u)\right)^{1/r}
	\end{equation}
	where $\mathcal{V}_{p}$ denotes the normalized Haar measure in $\mathbb{S}^{p-1}$. The metrics $d_r$ and $\rho_r$ (for the same value of $r$) are equivalent.

	\subsection{Fuzzy random variables}\label{Frv}
	
	There exists different definitions of fuzzy random variables in the literature. Here we consider the Puri's and Ralescu's approach (see \citep{PuriRalescu}). Let $(\Omega,\mathcal{A},\mathbb{P})$ be a probability space. A
	\emph{random compact set} \citep{Mol} is a function $\Gamma:\Omega\rightarrow\mathcal{K}_{c}(\mathbb{R}^{p})$ such that $\{\omega\in\Omega : \Gamma(\omega)\cap K\neq\emptyset \}\in\mathcal{A}$ for each
	$K\in\mathcal{K}_{c}(\mathbb{R}^{p})$. 	A \emph{fuzzy random variable} \citep{PuriRalescu} is a function $\mathcal{X}:\Omega\rightarrow\mathcal{F}_{c}(\mathbb{R}^{p})$ such that 	$\mathcal{X}_{\alpha}(\omega)$ is a random
	compact set for all $\alpha\in[0,1]$, 	where the $\alpha$-level mapping $\mathcal{X}_{\alpha}:\Omega\rightarrow\mathcal{K}_{c}(\mathbb{R}^{p})$ is defined by 	$\mathcal{X}_{\alpha}(\omega) := \{x\in\mathbb{R}^{p}:
	\mathcal{X}(\omega)(x)\geq\alpha  \}$ for any $\omega\in\Omega$. 	
	
	It is not explicit in this definition that a fuzzy random variable is a measurable function in the ordinary sense. But clearly, $\mathcal{X}$ is a fuzzy
	random variable if and only if it is measurable when $\mathcal{F}_{c}(\mathbb{R}^{p})$ is endowed with the $\sigma$-algebra generated by the $\alpha$-cut mappings $L_\alpha:A\in\mathcal{F}_{c}(\mathbb{R}^{p})\mapsto
	A_\alpha\in\mathcal{K}_{c}(\mathbb{R}^{p})$, namely the smallest $\sigma$-algebra which makes each $L_\alpha$ measurable. As shown by Kr\"atschmer \citep{Kra01}, that is the Borel $\sigma$-algebra generated by any of the
	metrics $d_r$ or $\rho_r$ for $r\in[1,\infty)$.  Given a fuzzy random variable, $\mathcal{X}: \Omega\rightarrow\mathcal{F}_{c}(\mathbb{R}^{p})$,   the support function of $\mathcal{X}$ is defined as the function
	$s_{\mathcal{X}} : \mathbb{S}^{p-1}\times [0,1]\times\Omega\rightarrow\mathbb{R}$ with $ 	s_{\mathcal{X}}(u,\alpha,\omega) := s_{\mathcal{X}(\omega)}(u,\alpha), $ for all $u\in\mathbb{S}^{p-1}, \alpha\in [0,1]$ and
	$\omega\in\Omega$. Throughout the paper, $(\Omega,\mathcal{A},\mathbb{P})$ denotes the probabilistic space associated with the fuzzy random variable $\mathcal{X}$.  Let $L^{0}[\mathcal{F}_{c}(\mathbb{R}^{p})]$ denote the
	class of all fuzzy random variables on the measurable space $(\Omega,\mathcal{A})$ and $C^{0}[\mathcal{F}_{c}(\mathbb{R}^{p})]\subseteq L^{0}[\mathcal{F}_{c}(\mathbb{R}^{p})]$ the class of all fuzzy random variables
	$\mathcal{X}$ such that $s_{\mathcal{X}}(u,\alpha)$ is a continuous real random variable for each $(u,\alpha)\in\mathbb{S}^{p-1}\times [0,1]$.

	\subsection{Fuzzy symmetry and depth. Semilinear and geometric depth.}\label{notiondepth}

	Let $\mathcal{X}:\Omega\rightarrow\mathcal{F}_{c}(\mathbb{R}^{p})$ be a fuzzy random variable and $A\in\mathcal{F}_{c}(\mathbb{R}^{p})$ a fuzzy set. In \citep{primerarticulo}, we proposed the {\em F-symmetry} notion for fuzzy random variables: 
	$\mathcal{X}$ is \emph{$F$-symmetric} with respect to $A$ if, for all $(u,\alpha)\in\mathbb{S}^{p-1}\times[0,1],$
	\begin{equation}\nonumber 		s_{A}(u,\alpha) - s_{\mathcal{X}}(u,\alpha) =^{\mathcal{L}} s_{\mathcal{X}}(u,\alpha) - s_{A}(u,\alpha). 	\end{equation} It can be checked that the indicator function $I_{\{X\}}$ of a
	$p$-dimensional random vector $X$ is F-symmetric if and only if $X$ is a symmetrically distributed random vector.
	
	Let $\text{Med}$ be the (possibly multivalued) median operator on real random variables. It is also proved in \cite{primerarticulo} that, for all $u\in\mathbb{S}^{p-1}$ and $\alpha\in [0,1],$
	\begin{eqnarray}\label{Amedian}
		s_{A}(u,\alpha) \in \text{Med}(s_{\mathcal{X}}(u,\alpha)),  \mbox{ if }\mathcal{X}  \mbox{ is } F\mbox{-symmetric with respect to }A. 		\end{eqnarray} In the sequel, given a real sample $x_{1}, \ldots, x_{n},$
	$\text{Med}(x_{1}, \ldots, x_{n})$ denotes its median.

	Let $\mathcal H\subseteq L^0[\pfc(\R^p)],$  $\mathcal{J}\subseteq\mathcal{F}_{c}(\mathbb{R}^{p}),$ and $d:\mathcal{F}_{c}(\mathbb{R}^{p})\times\mathcal{F}_{c}(\mathbb{R}^{p})\rightarrow[0,\infty)$ a metric. The following
	properties are considered in \citep{primerarticulo}. In them, $A$ denotes an element of $\mathcal{J}$ such that  $D(A;\mathcal{X}) = \sup \{D(B;\mathcal{X}) : B\in\mathcal{J}\}$, i.e., a fuzzy set of maximal depth in the
	distribution of $\mathcal X$.
	
	\begin{enumerate}
		\item[{\bf P1.}] $D(M\cdot C + B; M\cdot\mathcal{X} + B) = D(C;\mathcal{X})$  for any non-sigular matrix $M\in\mathcal{M}_{p\times p}(\mathbb{R}),$ any $B,C\in\mathcal{J}$ and any $\mathcal{X}\in{\mathcal{H}}.$ 	
		\item[{\bf P2.}] For (some notion of symmetry and) any symmetric fuzzy random variable  $\mathcal{X}\in{\mathcal H}$, 	$D(U;\mathcal{X}) = \sup_{B\in\mathcal{F}_{c}(\mathbb{R}^{p})} D(B;\mathcal{X}),$ where
		$U\in\mathcal{J}$ is a center of symmetry of $\mathcal{X}.$ 	\item[{\bf P3a.}] 		$ 		D(A;\mathcal{X})\geq D((1-\lambda)\cdot A + \lambda\cdot B;\mathcal{X})\geq D(B;\mathcal{X}) 		$ 		for all
		$\lambda\in[0,1]$ and all $B\in\mathcal{F}_{c}(\mathbb{R}^{p})$. 	\item[{\bf P3b.}] $ 		D(A;\mathcal{X})\geq D(B;\mathcal{X})\geq D(C;\mathcal{X})$		for all $B,C\in\mathcal{J}$ satisfying $d(A,C) = d(A,B) +
		d(B,C)$. 	\item[{\bf P4a.}] $ 		\lim_{\lambda\rightarrow\infty} D(A + \lambda\cdot B;\mathcal{X}) = 0 		$ 	for all $B\in\mathcal{J}\setminus\{\text{I}_{\{0\}}\}$. 	\item[{\bf P4b.}] $ 		
		\lim_{n\rightarrow\infty} D(A_{n};\mathcal{X}) = 0$		for every sequence of fuzzy sets $\{ A_{n}\}_{n}$ such that the $\lim_{n\rightarrow\infty} d(A_{n},A) = \infty$. 	\end{enumerate} In Property {\bf P2}, F-symmetry
	will be considered. Another notion of symmetry is also proposed in \cite{primerarticulo}. According to \citep{primerarticulo}, a mapping  $D(\cdot;\cdot):\mathcal{J}\times 	{\mathcal H} 	\rightarrow[0,\infty)$  is
	a \emph{semilinear depth function} if it satisfies P1, P2, P3a and P4a for each fuzzy random variable $\mathcal{X}\in\mathcal H.$ It  is  a \emph{geometric depth function} with respect to $d$  if it satisfies 	 P1, P2,
	P3b and P4b for each fuzzy random variable $\mathcal{X}\in\mathcal H.$ Notice that semilinear depth only depends on the arithmetics of $\pfc(\R^p)$ while geometric depth depends on the choice of a specific metric.
	
	\section{Pseudosimplices in $\pfc(\R^d)$}
	\label{pseudo}

	One of the most well-known statistical depth functions for multivariate data is simplicial depth \citep{LiuSimplicial}. Simplicial depth is an instance of what Zuo and Serfling \citep{ZuoSerfling}  called `Type A depth',
	i.e., the depth of a point is the probability that it lies in a certain random set constructed from independent and identically distributed copies of the random variable. As such, it is the coverage function of
	a random set and a connection to fuzzy sets is immediate \citep{GooNgu}. Further examples of Type A depth functions are majority depth \citep{majority, Parelius}, convex hull peeling depth \cite{Barnet1976}, spherical depth
	\citep{spherical}, and lens depth \citep{lens}.
	
	The simplicial depth of $x\in\mathbb{R}^{p}$ with respect to a probability distribution $\mathbb{P}$ on $\mathbb{R}^{p}$ is defined to be
	\begin{equation}\label{simplicialmult}
		SD(x;\mathbb{P}) := \mathbb{P}(x\in S[X_{1},\ldots,X_{p+1}]),
	\end{equation}
	where $X_{1},\ldots,X_{p+1}$ are independent and identically distributed random variables with distribution $\mathbb{P}$ and,  for any $x_{1},\ldots,x_{p+1}\in\mathbb{R}^{p},$ $S[x_{1},\ldots,x_{p+1}]$ is the set
	\begin{equation}\label{S}
		S[x_{1},\ldots,x_{p+1}] := \{\lambda_{1}x_{1} + \ldots + \lambda_{p+1}x_{p+1} : \sum_{i=1}^{p+1}\lambda_{i} = 1, \lambda_{i}\geq 0\},
	\end{equation}
	i.e., $S[x_{1},\ldots,x_{p+1}]$ is the convex hull of the points $x_{1},\ldots,x_{p+1}.$
	A characterization of simplices in $\R^p$  is provided in the next result.
	
	\begin{proposition}\label{proposicionsimplicieproyeccion}
		For any $x_{1},\ldots,x_{p+1}\in\mathbb{R}^{p},$ 	\begin{equation}\nonumber 		S[x_{1},\ldots,x_{p+1}] = \{x\in\mathbb{R}^{p} : \langle u,x\rangle\in [m(u), M(u)] \text{ for all } u\in\mathbb{S}^{p-1}\}, 	
		\end{equation} 	with
		$m(u):= \min\{\langle u,x_{1}\rangle,\ldots,\langle
		u,x_{p+1}\rangle\}$ and $M(u) := \max\{\langle u,x_{1}\rangle,\ldots,\langle u,x_{p+1}\rangle\}.$ 	
	\end{proposition}
	If the $X_i$'s are affinely independent, $S[X_{1},\ldots,X_{p+1}]$ is by definition a (random) $p$-dimensional simplex, which explains the name `simplicial depth'. Indeed, the $X_i$'s are affinely independent, almost
	surely, provided that $\mathbb{P}$ assigns zero probability to any lower-dimensional subspace of $\R^p;$ which is the case for continuous distributions. In the statistical depth
	literature, the name `simplex' reflects the fact that exactly $p+1$ points are taken for the convex hull, although it can fail to be $p$-dimensional for an arbitrary distribution $\mathbb{P}$. With this in mind, we will
	freely call $S[X_{1},\ldots,X_{p+1}]$ a simplex in the sequel.
	
	% %\begin{remark} {\bf >Eliminar?} %Proposition \ref{proposicionsimplicieproyeccion} is trivial for  $p = 1.$ Let $x_{1},x_{2}\in\mathbb{R}$.  By (\ref{S}) we have that $S[x_{1},x_{2}] =
	%[\min\{x_{1},x_{2}\},\max\{x_{1},x_{2}\}]$. %The proof follows from %$\mathbb{S}^{0} = \{-1,1\}$ and the fact that $x\in [\min\{x_{1},x_{2}\},\max\{x_{1},x_{2}\}]$ if, and only if, $-x\in [-\max\{x_{1},x_{2}\},-\min\{x_{1},x_{2}\}]$. %\end{remark}
	
	Before proposing plausible fuzzy depth instances inspired by the simplicial depth, we study how to adapt simplices to our context. To the best of our knowledge, the literature contains no notion of a simplex in
	$\mathcal{F}_{c}(\mathbb{R}^{p})$.  In  \citep{Cascos}, however, a {\em band}  generated by compact and convex sets is defined, which coincides with our definition of a pseudosimplex in $\mathcal{K}_{c}(\mathbb{R}^{p})$
	(Definition \ref{SimplexCompact} below). We  analyze it first to later  make use of it in our proposed  definition of a pseudosimplex  in $\mathcal{F}_{c}(\mathbb{R}^{p})$.  The justification for using the definition in
	\citep{Cascos} is that, according to Proposition \ref{proposicionsimplicieproyeccion}, the simplex generated by $p+1$ points, $x_{1},\ldots,x_{p+1}$, coincides with the set of points whose projections in every direction
	$u\in\mathbb{S}^{p-1}$ are in the closed interval generated by the minimum  and the maximum of $\langle u, x_{1}\rangle ,\ldots,\langle u, x_{p+1}\rangle$. Thus, replacing in this characterization the inner products by the
	support function of the elements in $\mathcal{K}_{c}(\mathbb{R}^{p})$  yields the following definition.
	
	\begin{definition}\label{SimplexCompact}
		The {\em pseudosimplex} generated by $A_{1},\ldots,A_{p+1}\in\mathcal{K}_{c}(\mathbb{R}^{p})$ is 	\begin{equation}\nonumber 		S_c[A_{1},\ldots,A_{p+1}]  := \{A\in\mathcal{K}_{c}(\mathbb{R}^{p}) : s_{A}(u)\in [m(u),
			M(u)] \text{ for all } u\in\mathbb{S}^{p-1}\}, 	\end{equation} 	where $m(u) := \min\{s_{A_{1}}(u),\ldots,s_{A_{p+1}}(u)\}$ and $M(u) := \max\{s_{A_{1}}(u),\ldots,s_{A_{p+1}}(u)\}.$
	\end{definition}
	As simplices are defined to be subsets of linear spaces and $\pkc(\R^d)$ and $\pfc(\R^d)$ are not linear but they embed into appropriate linear spaces (e.g., by identifying their elements with support functions), there
	arises the question whether, after such an embedding, $S_c[A_1,\ldots, A_{p+1}]$ becomes an infinite-dimensional simplex \cite[Section 1.5, pp. 46--53]{ChoquetOrder}. The name `pseudosimplex' avoids prejudicing the question.

	As the operations of sum and product by a scalar are defined in  $\mathcal{K}_{c}(\mathbb{R}^{p})$ (Section \ref{Aarith}) 	 an alternative could be to  define the simplex generated by
	$A_{1},\ldots,A_{p+1}\in\mathcal{K}_{c}(\mathbb{R}^{p})$ as the set of all convex combinations of these generating elements, that is 	\begin{equation}\label{SimplexCompact2} 		\{A\in\mathcal{K}_{c}(\mathbb{R}^{p}) :
		A = \sum_{i = 1}^{p+1}\lambda_{i}\cdot A_{i}, \mbox{ with }  \sum_{i = 1}^{p+1}\lambda_{i} = 1 \mbox{ and } \lambda_{i}\geq 0 \}. 	\end{equation} That corresponds to the convex hull of the set $\{A_{1},\ldots,A_{p+1}\}$
	when $\pkc(\R^d)$ is regarded as a convex combination space \citep{JTP}. The next result proves that every simplex in the sense of  \eqref{SimplexCompact2} is contained in the corresponding pseudosimplex.  Example
	\ref{ExampleSimplices} shows that both sets are not necessarily equal. 
	\begin{proposition}\label{PropositionSimplices}
		For any $A_{1},\ldots,A_{p+1}\in\mathcal{K}_{c}(\mathbb{R}^{p}),$
		$$\{A\in\mathcal{K}_{c}(\mathbb{R}^{p}) : A = \sum_{i = 1}^{p+1}\lambda_{i}\cdot A_{i}, \sum_{i = 1}^{p+1}\lambda_{i} = 1, \lambda_{i}\geq 0
		\}\subseteq S_{c}[A_{1},\ldots,A_{p+1}].$$
	\end{proposition}
	
	\begin{example}\label{ExampleSimplices}
		Let $p = 1,$ $A = [0,1]$ and $B = [3,4].$ Then 	\begin{equation}\nonumber 		S_c[A,B] = \{[x,y] : x\in [0,3], y\in [1,4]\} 	\end{equation} 	while the simplex in the sense of Equation \eqref{SimplexCompact2} is 	
		\begin{equation}\nonumber 		S := \{[3 \lambda, 1 + 3 \lambda] : \lambda\in [0,1]\}.
		\end{equation} 	
		For instance, $\{2\}\in S_c[A,B]$ but $\{2\}\not\in S$.
	\end{example}
	The choice of the pseudosimplex, instead of the convex hull simplex in \eqref{SimplexCompact2}, is based on cases like the last example. Intuitively, it is hard to deny that $\{2\}$ is between $A$ and $B$ in a definite
	sense, but it cannot be written as a convex combination of them. In this connection, see  Proposition \ref{aaa} below concerning the role of `betweenness' in the definition of pseudosimplices in the fuzzy case.

	We will extend now the notion of a pseudosimplex to the fuzzy case by working $\alpha$-level by $\alpha$-level.
	
	\begin{definition}\label{pseudosimplex}
		The {\em pseudosimplex} 	 generated by $A_{1},\ldots,A_{p+1}\in\mathcal{F}_{c}(\mathbb{R}^{p})$ is 	\begin{equation}\nonumber 	S_{F}[A_{1},\ldots,A_{p+1}] := \{A\in\mathcal{F}_{c}(\mathbb{R}^{p}) : A_{\alpha}\in
			S_c[(A_{1})_{\alpha},\ldots,(A_{p+1})_{\alpha}] \text{ for all } \alpha\in [0,1]\}, 	\end{equation} 	where $(A_{i})_{\alpha}$ denotes the $\alpha$-level of  $A_{i}.$ % $i\in\{1,\ldots,p+1\}$ and  $\alpha\in [0,1]$.
	\end{definition}
	As fuzzy sets are a generalization of ordinary sets in $\mathbb{R}^{p}$, it is interesting to underline that the notion of a pseudosimplex generated by crisp sets contains that of a simplex in the multivariate case. For
	that, we consider the class of fuzzy sets
	$$\mathcal{R}^{p} := \{\text{I}_{\{x\}}\in\mathcal{F}_{c}(\mathbb{R}^{p}): x\in\mathbb{R}^{p}\},$$
	which can be identified with $\mathbb{R}^{p}$ (Section \ref{Aarith}). %
	\begin{proposition}\label{zzz}
		For any $x_{1},\ldots,x_{p+1}\in\mathbb{R}^{p},$ 	\begin{equation}\nonumber 		S_{F}[\text{I}_{\{x_{1}\}},\ldots,\text{I}_{\{x_{p+1}\}}]\cap\mathcal{R}^{p} = \{\text{I}_{\{x\}} : x\in S[x_{1},\ldots,x_{p+1}]\}. 	
		\end{equation}
	\end{proposition}
	The proof of the result is trivial. A direct implication of the proposition  is
	\begin{equation}\label{C1}
		\{\text{I}_{\{x\}} : x\in S[x_{1},\ldots,x_{p+1}]\}\subseteq S_{F}[\text{I}_{\{x_{1}\}},\ldots,\text{I}_{\{x_{p+1}\}}].
	\end{equation}
	Furthermore,
	\begin{equation}\label{C2}
		\{\text{I}_{\{x\}} : x\in S[x_{1},\ldots,x_{p+1}]\}\subsetneq S_{F}[\text{I}_{\{x_{1}\}},\ldots,\text{I}_{\{x_{p+1}\}}]
	\end{equation}
	provided there exist $i,j\in\{{1},\ldots, {p+1}\}$ such that $x_i\neq x_j$.
	%This is because of the following.
	As $S[x_{1},\ldots,x_{p+1}]$ is a convex set, it contains the segment joining $x_{i}$ and $x_{j}.$  Denoting it  by $\overline{x_{i}x_{j}},$   we have  $$\text{I}_{\overline{x_{i}x_{j}}}\in
	S_{F}[\text{I}_{\{x_{1}\}},\ldots,\text{I}_{\{x_{p+1}\}}].$$ However, $$\text{I}_{\overline{x_{i}x_{j}}}\not\in\{\text{I}_{\{x\}} : x\in S[x_{1},\ldots,x_{p+1}]\}$$ because $\overline{x_{i}x_{j}}$ is not a single point.
	
	Another corollary is that the result in Proposition \ref{zzz} is also obtained for $\mathcal{K}_{c}(\mathbb{R}^{p}).$ For that, we denote  $\mathcal{R}_{c}^{p} := \{\{x\}\in\mathcal{K}_{c}(\mathbb{R}^{p}) :
	x\in\mathbb{R}^{p}\},$ the set of singletons. 	\begin{corollary}\label{corolariozzz} 		For any $x_{1},\ldots,x_{p+1}\in\mathbb{R}^{p},$ we have that 		\begin{equation}\nonumber 			
			S_{c}[\{x_{1}\},\ldots,\{x_{p+1}\}]\cap\mathcal{R}_{c}^{p} = \{\{x\} : x\in S[x_{1},\ldots,x_{p+1}]\}. 		\end{equation} 	\end{corollary}
	We also have the  inclusions in \eqref{C1} and \eqref{C2}  for this particular case. %An implication from Corollary \ref{corolariozzz} would be that of crisp sets, in which  the pseudosimplex is generated by sets in $\pkc(\R)$, see Definition \ref{SimplexCompact}. 
	An example is that of the pseudosimplex generated by $\{0\}$ and $\{3\},$ which  contains not only singletons but also sets like the interval $[1,2]$ which lies entirely in
	the gap between 0 and 3.

	The Ram\'\i k--\v R\'\i man\'ek partial order in $\pfc(\R)$ \cite[Definition 3]{RR} is given by
	$$A_1\preceq A_2 \Leftrightarrow \inf (A_1)_\alpha\le \inf (A_2)_\alpha, \sup (A_1)_\alpha\le \sup (A_2)_\alpha \;\;\forall \alpha\in (0,1].$$
	This provides a natural (partial) ordering in $\pfc(\R),$ which ranking methods for fuzzy numbers should  be consistent with.
	
	\begin{proposition}\label{aaa}
		Let $A_1,A_2\in\pfc(\R)$. If $A_1\preceq A_2$ then $S_F[A_1,A_2]$ is  the set of all $A\in\pfc(\R^p)$ such that $A_1\preceq A\preceq A_2$.
	\end{proposition}
	
	Propositions \ref{zzz} and \ref{aaa} confirm that pseudosimplices are consistent with a natural notion of `being between' for fuzzy numbers; as opposed to what would have happened with convex hull simplices.
	
	\section{Simplicial depths for fuzzy sets}
	\label{simplicial}

	Our constructions of an analog to simplicial depth are not the direct result of plugging the fuzzy pseudosimplex into the simplicial depth formula. To
	understand why, we first propose and discuss a more straightforward adaptation.
	
	The \emph{naive simplicial depth}, based on $\mathcal{J}\subseteq\mathcal{F}_{c}(\mathbb{R}^{p})$ and $ \mathcal{H}\subseteq L^{0}[\mathcal{F}_{c}(\mathbb{R}^{p})]$, of a fuzzy set $A\in\mathcal{J}$ with respect to a fuzzy
	random variable  $\mathcal{X}\in\mathcal{H}$ is %given by the function $D_{nS} : \mathcal{J}\times\mathcal{H}\rightarrow [0,1]$ 
	\begin{equation}\label{DnS} 		D_{nS}(A;\mathcal{X}) := \mathbb{P}(A\in S_{F}[\mathcal{X}_{1},\ldots,\mathcal{X}_{p+1}]), 	\end{equation} 	where  $\mathcal{X}_{1},\ldots,\mathcal{X}_{p+1}$ are $p+1$ independent
	observations. 	Setting 	\begin{eqnarray}\label{m}m_{\mathcal{X}}(u,\alpha) := \min\{s_{\mathcal{X}_{1}}(u,\alpha),\ldots,s_{\mathcal{X}_{p+1}}(u,\alpha)\},\\
		\label{Mu}M_{\mathcal{X}}(u,\alpha) := \max\{s_{\mathcal{X}_{1}}(u,\alpha),\ldots,s_{\mathcal{X}_{p+1}}(u,\alpha)\},
	\end{eqnarray}
	for  any $(u,\alpha)\in\mathbb{S}^{p-1}\times [0,1],$ we can also express this function as 	
	\begin{equation}\label{expressionSDF1} 		
		D_{nS}(A;\mathcal{X}) = \mathbb{P}\left(s_{A}(u,\alpha)\in [m_{\mathcal{X}}(u,\alpha),
		M_{\mathcal{X}}(u,\alpha)] \text{ for all }(u,\alpha)\in\mathbb{S}^{p-1}\times [0,1] \right). 	
	\end{equation} 	
	
	It is not self-evident that $D_{nS}$ is well defined:
	\begin{itemize}
		\item[(i)] In (\ref{DnS}), it is not clear whether $S_{F}[\mathcal{X}_{1},\ldots,\mathcal{X}_{p+1}]$ is a random set in $\pfc(\R^p)$, which would ensure that the probability makes sense.
		\item[(ii)] In (\ref{expressionSDF1}), the event depends on uncountably many $(u,\alpha)$, making it an uncountable intersection which might fail to be measurable.
	\end{itemize}
	Thus it becomes necessary to establish the measurability of those events. The proof that $D_{nS}$ is well defined, as is the case with the simplicial depths in the sequel, is presented in Section {\color{red}\ref{proofs}}.

	The proposed naive simplicial fuzzy  depth  generalizes  the multivariate simplicial depth, as observed below by taking  $\mathcal{J} = \mathcal{R}^{p}.$		 	
	\begin{proposition}
		\label{simplicialRp} 		  For any random variable $X$ on $\mathbb{R}^{p}$  and $x\in\mathbb{R}^{p},$ $$D_{nS}(\text{I}_{\{x\}};\text{I}_{X}) = SD(x;\mathbb P_X).$$
	\end{proposition}
	The proof follows directly.  Although we replaced convex hull simplices by pseudosimplices, which are generally larger, this naive depth function may still result in a high number of ties at zero, which is inappropriate for
	certain applications such as classification. That is a consequence of the fuzzy set having to be completely contained in the pseudosimplex. An analogous problem was observed by L\'opez-Pintado and Romo when adapting
	simplicial depth to functional data in \citep{LopezRomoBand}. Their definition of {\em band depth} is aimed at ordering functional data and stems from the simplicial depth in the same way as our naive simplicial fuzzy
	depth. 	To overcome this shortcoming, in \citep{LopezRomoBand} a {\em modified band depth} is introduced which inspires our next definition. A similar reasoning is also found in \citep{halfregion,SimplicialGenton}, both in
	the functional setting.
	
	\begin{definition}\label{DmS}
		The \emph{modified simplicial depth},  based on $\mathcal{J}\subseteq\mathcal{F}_{c}(\mathbb{R}^{p})$ and $ \mathcal{H}\subseteq L^{0}[\mathcal{F}_{c}(\mathbb{R}^{p})]$, of a fuzzy set $A\in\mathcal{J}$ with respect to a
		random variable  $\mathcal{X}\in\mathcal{H}$ is 	\begin{equation}\nonumber 		D_{mS}(A;\mathcal{X}) := \text{E}(\mathcal{V}_{p}\otimes\nu\{(u,\alpha)\in\mathbb{S}^{p-1}\times [0,1] : s_{A}(u,\alpha)\in
			[m_{\mathcal{X}}(u,\alpha), M_{\mathcal{X}}(u,\alpha)]\}), 	\end{equation} 	where $m_{\mathcal{X}}(u,\alpha)$ and $M_{\mathcal{X}}(u,\alpha)$ are defined in  \eqref{m} and  \eqref{Mu} and
		$\mathcal{X}_{1},\ldots,\mathcal{X}_{p+1}$ 	 are independent observations of $\mathcal{X}.$
	\end{definition}
	By Fubini's Theorem (see Section \ref{proofs} for a detailed justification),
	\begin{equation}\label{ecuacionmSDFp}
		D_{mS}(A;\mathcal{X}) = \int_{\mathbb{S}^{p-1}}\int_{[0,1]} \mathbb{P}\left(s_{A}(u,\alpha)\in [m_{\mathcal{X}}(u,\alpha),M_{\mathcal{X}}(u,\alpha)]\right)  \dif\nu(\alpha) \dif\mathcal{V}_{p}(u).
	\end{equation}
	This inspires us to introduce the following definition of  simplicial fuzzy depth, which is also motivated by the Tukey depth in \cite{primerarticulo} (which is defined as an infimum over $\mathbb{S}^{p-1}$).
	
	\begin{definition}\label{DFS}
		The \emph{simplicial depth}  based on $\mathcal{J}\subseteq\mathcal{F}_{c}(\mathbb{R}^{p})$ and $ \mathcal{H}\subseteq L^{0}[\mathcal{F}_{c}(\mathbb{R}^{p})]$ of a fuzzy set $A\in\mathcal{J}$ with respect to a
		random variable  $\mathcal{X}\in\mathcal{H}$ is 	\begin{equation}\nonumber 		D_{FS}(A;\mathcal{X}) := \inf_{u\in\mathbb{S}^{p-1}}\text{E}(\nu\{\alpha\in [0,1] : s_{A}(u,\alpha)\in [m_{\mathcal{X}}(u,\alpha),
			M_{\mathcal{X}}(u,\alpha)]\}), 	\end{equation} 	where $m_{\mathcal{X}}(u,\alpha)$ and $M_{\mathcal{X}}(u,\alpha)$ are defined in  \eqref{m} and  \eqref{Mu} and $\mathcal{X}_{1},\ldots,\mathcal{X}_{p+1}$ are independent
		observations of $\mathcal{X}$.
	\end{definition}
	Again by Fubini's Theorem,
	\begin{equation}\label{ecuacionMSDFp}
		D_{FS}(A;\mathcal{X}) = \inf_{u\in\mathbb{S}^{p-1}}\int_{[0,1]} \mathbb{P}(s_{A}(u,\alpha)\in [m_{\mathcal{X}}(u,\alpha),M_{\mathcal{X}}(u,\alpha)]) \dif\nu(\alpha).
	\end{equation}
	
	The difference between both definitions could be understood in the following way. In  \eqref{ecuacionmSDFp} we take the average over $\mathbb{S}^{p-1}$ of the integral over $[0,1]$, while in  \eqref{ecuacionMSDFp} we take
	the infimum over $\mathbb{S}^{p-1}$ of the integral over $[0,1]$, that is, we consider the direction $u\in\mathbb{S}^{p-1}$ where the integral over $[0,1]$ is smallest.
	The next example shows the difference between
	$D_{mS}$ and $D_{FS},$ and their suitability under distinct scenarios. The example is in $\mathcal{F}_{c}(\mathbb{R}),$ in which the expressions in Definitions \ref{DmS} and \ref{DFS} reduce to 	\begin{equation}\label{em}
		D_{mS}(A;\mathcal{X}) = \displaystyle\cfrac{1}{2}\sum_{u\in\{-1,1\}}\text{E}(\nu\{\alpha\in [0,1] : s_{A}(u)\in [m_{\mathcal{X}}(u,\alpha),M_{\mathcal{X}}(u,\alpha)]\})
	\end{equation}
	and 	\begin{equation}\label{ed} D_{FS}(A;\mathcal{X}) = \min_{u\in\{-1,1\}}\{\text{E}(\nu\{\alpha\in [0,1] : s_{A}(u)\in [m_{\mathcal{X}}(u,\alpha),M_{\mathcal{X}}(u,\alpha)]\})\}.
	\end{equation}
	
	\begin{example}\label{Exfig}
		Let $(\{\omega_{1},\omega_{2}\},\mathcal{P}(\{\omega_{1},\omega_{2}\}),\mathbb{P})$ be a probabilistic space with $\mathbb{P}(\omega_{1}) = \mathbb{P}(\omega_{2})$. We consider the fuzzy random variable $$\mathcal{X} :
		\{\omega_{1},\omega_{2}\}\rightarrow\mathcal{F}_{c}(\mathbb{R}) \mbox{ defined by } \mathcal{X}(\omega_{1}) = \text{I}_{[1,2]} \mbox{ and } \mathcal{X}(\omega_{2}) = \text{I}_{[4,5]}.$$ Let $\mathcal{X}_1, \mathcal{X}_2$ be
		two independent observations of $\mathcal{X}$ such that $\mathcal{X}_i=\mathcal{X}(\omega_{i}),$ for $i=1,2.$  With this,  for each $\alpha\in [0,1],$ we have that $$s_{\mathcal{X}_{1}}(-1,\alpha) = -1\mbox{, }
		s_{\mathcal{X}_{1}}(1,\alpha) = 2\mbox{, } s_{\mathcal{X}_{2}}(-1,\alpha) = -4\mbox{ and } s_{\mathcal{X}_{2}}(1,\alpha) = 5.$$ 	
		Then,  in the present example, the expressions in \eqref{em} and \eqref{ed} for a general
		$A\in\mathcal{F}_{c}(\mathbb{R})$ result in
		\begin{equation}\label{eem}
			D_{mS}(A;\mathcal{X}) = \displaystyle\cfrac{1}{2}[\nu\{\alpha\in [0,1] : s_{A}(1,\alpha)\in [2,5]\}  + \nu\{\alpha\in [0,1] : s_{A}(-1,\alpha)\in [-4,-1]\}]
		\end{equation}
		and
		\begin{equation}\label{eef}
			D_{FS}(A;\mathcal{X}) = \min\{\nu\{\alpha\in [0,1] : s_{A}(1,\alpha)\in [2,5]\} , \nu\{\{\alpha\in [0,1] : s_{A}(-1,\alpha)\in [-4,-1]\}\}.
		\end{equation}
		We propose two cases:
		\begin{itemize}
			\item[(i)]  $R,G\in\mathcal{F}_{c}(\mathbb{R}^{p})$ such that $D_{mS}(R;\mathcal{X}) = D_{mS}(G;\mathcal{X})$ and $D_{FS}(R;\mathcal{X}) \neq D_{FS}(G;\mathcal{X});$ 		\item[(ii)]
			$R,G\in\mathcal{F}_{c}(\mathbb{R}^{p})$ such that $D_{mS}(R;\mathcal{X}) \neq D_{mS}(G;\mathcal{X})$ and $D_{FS}(R;\mathcal{X}) = D_{FS}(G;\mathcal{X})$. 		\end{itemize}
		The  example is illustrated in Figure
		\ref{Fa}, case (i) in the top row and case (ii) in the bottom row. There, the $\mathcal{X}_i,$  $i=1,2,$ are represented in black and, in each case, $R$ in red and $G$  in green.
		
		\begin{figure}[!htbp]
			\begin{center} 		\includegraphics[width=0.49\linewidth]{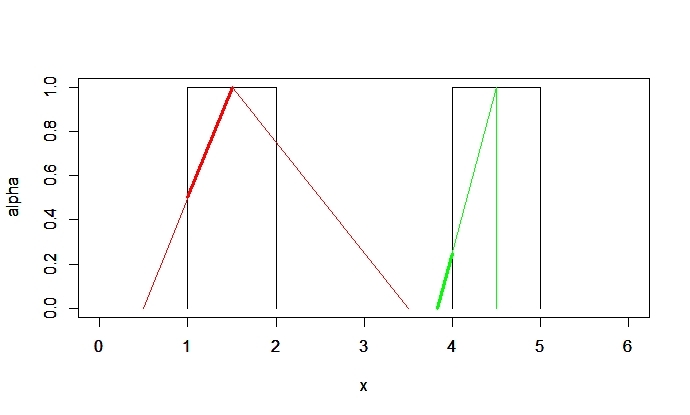} 		
				\includegraphics[width=0.49\linewidth]{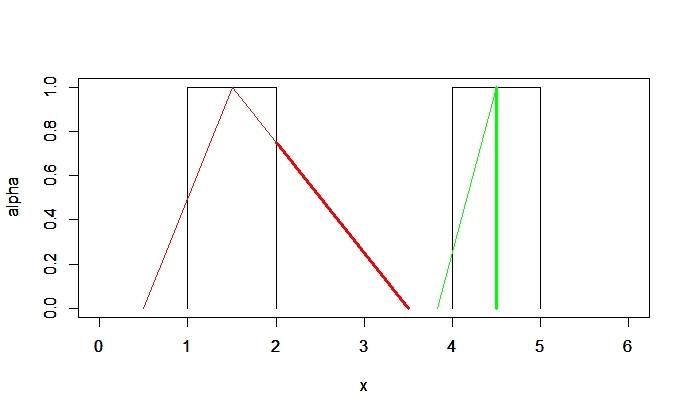} 		
				\includegraphics[width=0.49\linewidth]{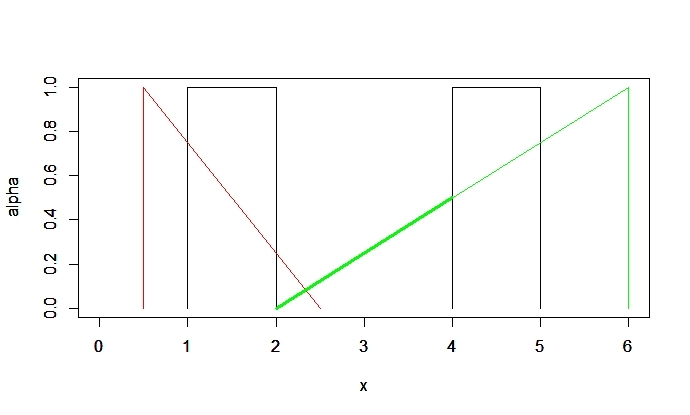} 		
				\includegraphics[width=0.49\linewidth]{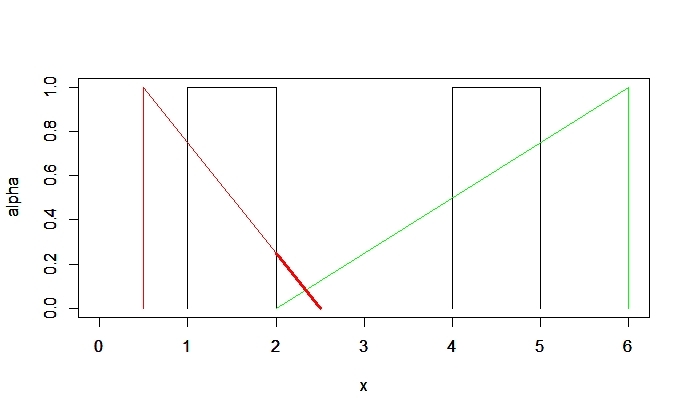} 	\end{center} 	\caption{Representation of Example \ref{Exfig}, part (i) in the top row and part (ii) in the bottom row. In each plot,  the fuzzy sets $\mathcal{X}_i$
				($i=1,2$) are represented in black, $R$ in red and $G$ in green. 	Thick lines indicate the parts of  $R$ and $G$  for which the corresponding support function is in the interval
				$[m_{\mathcal{X}}(u,\alpha),M_{\mathcal{X}}(u,\alpha)],$ with $u = -1$ in the left column  and  $u = 1$ in the right column. 	} 	\label{Fa}
		\end{figure}

		\begin{itemize}
			\item[(i)] Let $R, G\in\mathcal{F}_{c}(\mathbb{R}^{p})$ be defined, for any $t\in\mathbb{R},$ by 	\begin{eqnarray*}R(t) :=( t - 1/2)I_{ [1/2, 3/2]}(t)+( -t/2 + 7/4)I_{[3/2, 7/2]}(t)\\G(t) := (3t/2 - 23/4)I_{[23/6,
					9/2]}(t),\end{eqnarray*} 	Consequently,  $R_{\alpha} = [\alpha + 1/2, 7/2 - 2\cdot\alpha]$ and $G_{\alpha} = [(2/3)\cdot\alpha + 23/6, 9/2],$ $\alpha\in [0,1],$ are their $\alpha$-levels  and, for each $\alpha\in
			[0,1],$ \begin{eqnarray*}s_{R}(-1,\alpha) = -\alpha - 1/2\mbox{, } s_{R}(1,\alpha) = 7/2 - 2\cdot \alpha, \\ s_{G}(-1,\alpha) = -(2/3)\cdot\alpha - 23/6\mbox{ and } s_{G}(1,\alpha) = 9/2\end{eqnarray*}  are  their
			support functions.		
			
			To obtain the depth values, we first compute the Lebesgue measures of the $\alpha$'s for which these support functions belong to the intervals established in  \eqref{eem} and \eqref{eef}. We
			illustrate the computation with 	 the top row of Figure \ref{Fa}. In the left plot, the thick red line is  the part of the set $R$ for which $s_{R}(-1,\alpha)\in [-4,-1].$ This corresponds to $\alpha\in[.5,1]$ which
			results in a Lebesgue measure of 0.5. Meanwhile, in the right plot of Figure \ref{Fa}, the thick  red line is  the part of $R$ such that $s_{R}(1,\alpha)\in [2,5],$ which corresponds to $\alpha\in[0,.75],$ with Lebesgue
			measure .75. These measures add up to 5/4 with their minimum being 1/2. 
			
			Analogously, the thick green line in the left plot is
			the part of set $G$ for which $s_{G}(-1,\alpha)\in [-4,-1].$ This corresponds to $\alpha\in[0 , .25]$ which results in a Lebesgue measure of 0.25. In the right plot,
			the thick green line is the part of  $G$ such that $s_{G}(1,\alpha)\in [2,5].$ It corresponds to $\alpha\in[0 ,1],$ which results in a Lebesgue measure 1. These measures add again up to  5/4 but this time their minimum is 1/4.
			Thus, making use of  \eqref{eem} and \eqref{eef}, $$D_{mS}(R;\mathcal{X}) = D_{mS}(G;\mathcal{X}) = 5/8 \mbox{ and } D_{FS}(R;\mathcal{X}) = 1/2\neq 1/4 = D_{FS}(G;\mathcal{X}).$$ 	
			\item[(ii)] Let $R,
			G\in\mathcal{F}_{c}(\mathbb{R}^{p})$ be defined, for any $t\in\mathbb{R},$ by 	\begin{eqnarray*}R(t) := (-t/2 + 5/4)I_{ [1/2, 5/2]}(t) \\ G(t) := (t/4 - 1/2)I_{ [2, 6]}(t).\end{eqnarray*}   The corresponding
			$\alpha$-levels are $R_{\alpha} = [1/2, 5/2 - 2\cdot\alpha]$ and  $G_{\alpha} = [4\cdot\alpha + 2 , 6],$ $\alpha\in [0,1]$. Thus, for each $\alpha\in [0,1],$ we have the support functions $$s_{R}(-1,\alpha) = -
			1/2\mbox{, } s_{R}(1,\alpha) = 5/2 - 2\cdot \alpha\mbox{, } s_{G}(-1,\alpha) = -4\cdot\alpha - 2\mbox{ and } s_{G}(1,\alpha) = 6.$$ 	
			As in the previous case, we compute the Lebesgue measures of the $\alpha$'s for
			which these support functions belong to the intervals established in  \eqref{eem} and \eqref{eef}. This time, we clarify the computation making use of 	 the bottom row of Figure \ref{Fa}. As we can observe in the left
			plot, this time there is no thick red line, meaning that $s_{R}(-1,\alpha)\notin [-4,-1];$ and consequently the associated Lebesgue measure is 0. There is, however, a thick red line in the right plot, which coincides
			with  the part of $R$ such that $s_{R}(1,\alpha)\in [2,5].$ This corresponds to $\alpha\in[0 , .25],$ with a Lebesgue measure of .25. 	 For $G,$ things are kind of opposed.  $s_{G}(1,\alpha)\notin [2,5],$ which results in
			a 0 Lebesgue measure, and no thick green line in the bottom right plot of Figure \ref{Fa}. This time, for each  $\alpha\in[0 , .5]$ it is satisfied  that $s_{G}(-1,\alpha)\in [-4,-1].$ This results in a larger thick
			green line in the bottom left plot that results in a Lebesgue measure of .5. Thus, for $R$ and $G$ the minimum Lebesgue measure is 0. 	
			Taking into account  \eqref{ecuacionmSDFp}
			and \eqref{ecuacionMSDFp}, $$D_{mS}(R;\mathcal{X}) = 1/8 \neq 1/4 = D_{mS}(G;\mathcal{X}) \mbox{ and } D_{FS}(R;\mathcal{X}) =0 = D_{FS}(G;\mathcal{X}).$$
		\end{itemize}
	\end{example}
	
	This example shows the relevant differences and similarities  between $D_{mS}$ and $D_{FS}$. Let us comment them further, making use of the plots in   Figure \ref{Fa}. Focussing on case (i), top row plots, we have that $R$
	and $G$ take the same $D_{mS}$ depth value because the average  of the amount of $\alpha$'s corresponding to the thick red lines between the two plots is the same as the average corresponding to the thick green lines.
	However, non of those amounts is the same, which is depicted by $D_{FS},$ providing different depth values. It gives smaller depth value to $G$ because the amount of  $\alpha$'s corresponding to one of the thick green lines
	is the smallest among the four. In case (ii), bottom row plots, we have that $R$ and $G$ take the same $D_{FS}$ value because $R$ and $G$ result both in only one thick line each. $D_{mS}$ is able to depict a difference
	between $R$ and $G:$ that the thick line associated  to $G$ is larger than that associated to $R;$ giving then a higher depth value to $G.$ As commented before, the difference is due to the distinct way in which they
	summarize the information. One can argue that $D_{mS}$ is potentially better because it uses more information by computing the average. On the other hand, it can also be argued that $D_{FS}$ will extract the relevant
	information in certain problems.
	
	\section{Properties of $D_{mS}$, $D_{FS}$,  and $D_{nS}$}
	\label{properties}
	
	In this section, we will study whether the adaptations of simplicial depth to the fuzzy setting are semilinear and geometric depth functions in the sense of \citep{primerarticulo}.

	Theorem \ref{teoremamSDp} collects properties of the simplicial depth functions $D_{mS}$ and $D_{FS}$. Its proof is based on proofs of the simplicial band depth \cite[Theorems $1$ and $2$]{SimplicialGenton} and Proposition
	\ref{Tecuacion2MSDp}. The result is valid for $\mathcal H\subset C^{0}[\mathcal{F}_{c}(\mathbb{R}^{p})]$, namely fuzzy random variables all whose support functionals are continuous random variables.
	Note that, in order to define directly a notion of continuous fuzzy random variables, one would need first a reference measure with respect to which those variables would have a density function. In absence of such a
	measure (which would play the role of the Lebesgue measure in $\R^p$), the reduction to real random variables via the support function is more operative.

	\begin{proposition}\label{Tecuacion2MSDp}
		Let $\mathcal{X}\in L^{0}[\mathcal{F}_{c}(\mathbb{R}^{p})],$ $U\in\mathcal{F}_{c}(\mathbb{R}^{p})$  and $(u,\alpha)\in\mathbb{S}^{p-1}\times [0,1].$ Let $F_{u,\alpha}$ be the cumulative distribution function of the real
		random variable $s_{\mathcal{X}}(u,\alpha).$ Then 	\begin{equation*} 	\begin{aligned}
				\mathbb{P}(s_{U}(u,\alpha)\in [m_{\mathcal{X}}(u,\alpha), M_{\mathcal{X}}(u,\alpha&)]) = 1 - [1 - F_{u,\alpha}(s_{U}(u,\alpha))]^{p+1} \\ & \hspace{1cm}- [F_{u,\alpha}(s_{U}(u,\alpha))-\mathbb{P}(s_{\mathcal{X}}(u,\alpha) =
				s_{U}(u,\alpha))]^{p+1}.
			\end{aligned}
		\end{equation*} 	If, additionally, $\mathcal{X}\in C^{0}[\mathcal{F}_{c}(\mathbb{R}^{p})],$ that reduces to 		\begin{equation*}
			\mathbb{P}(s_{U}(u,\alpha)\in [m_{\mathcal{X}}(u,\alpha), M_{\mathcal{X}}(u,\alpha)]) = 1 - [1 - F_{u,\alpha}(s_{U}(u,\alpha))]^{p+1} - [F_{u,\alpha}(s_{U}(u,\alpha))]^{p+1}.
		\end{equation*}
	\end{proposition}
	
	\begin{theorem}\label{teoremamSDp}
		\label{teoremaMSDp}
		When computed with respect to an $F$-symmetric random variable $\mathcal{X}\in C^{0}[\mathcal{F}_{c}(\mathbb{R}^{p})]$,   $D_{mS}(\cdot;\mathcal{X})$ and $D_{FS}(\cdot;\mathcal{X})$ satisfy P1, P2, P3a and P3b for the $\rho_{r}$ distances for any $r\in (1,\infty)$. 
	\end{theorem}
	%\begin{remark}The hypothesis are used in proving P2 and P3a\end{remark}
	
	In general, $D_{mS}$ and $D_{FS}$  violate P4a, as shown by the following example. They also violate P4b, since P4b implies P4a \cite[Proposition $5.8$]{primerarticulo}.
	\begin{example}\label{ejemplosimplicialP4}
		Let $(\{\omega_{1},\omega_{2}\}, \mathcal{P}(\{\omega_{1},\omega_{2}\}),\mathbb{P})$ be a probability space such that $\mathbb{P}(\omega_{1}) = \mathbb{P}(\omega_{2})$ and   $$\mathcal{X} :
		\{\omega_{1},\omega_{2}\}\rightarrow\mathcal{F}_{c}(\mathbb{R}) \mbox{ with } \mathcal{X}(\omega_{1}) = \text{I}_{\{1\}} \mbox{ and } \mathcal{X}(\omega_{2}) = \text{I}_{\{-1\}}.$$
		It is clear that $\mathcal{X}$ is $F$-symmetric with respect to $A = \text{I}_{\{0\}}$. Let
		$B\in\mathcal{F}_{c}(\mathbb{R})$ such that, for any $t\in\mathbb{R},$ 	\begin{eqnarray*}B(t) := (-t/2 + 1/2)I_{(0,1]}(t)+I_{\{0\}}(t). 	\end{eqnarray*} 	 Thus, we have that $$B_{\alpha} = [0, 1-2\alpha] \mbox{ for
		}\alpha\in [0,1/2] \mbox{  and }B_{\alpha} = \{0\} \mbox{ for any } \alpha\in [1/2,1].$$ Additionally, $$s_{B}(-1,\alpha) = 0\mbox{ for all }\alpha\in [0,1]\mbox{ and }s_{B}(1,\alpha) = 0\mbox{ for all }\alpha\in [1/2,1].$$
		Taking into account the definition of $D_{FS}$, we have that,  for all $n\in \mathbb{N},$ $$D_{FS}(A + n\cdot B;\mathcal{X}) \geq 1/2; \mbox{   consequently, } \lim_{n\rightarrow\infty} D_{FS}(A + n\cdot B;\mathcal{X}) >
		0.$$	Analogously, we have that $D_{mS}(A+n\cdot B;\mathcal{X})\geq 1/2$ for all $n\in\mathbb{N}$, thus $$\lim_{n\rightarrow\infty}D_{mS}(A+n\cdot B;\mathcal{X}) > 0.$$
	\end{example}
	In property P4a we study sequences of fuzzy sets of the form $\{A + n\cdot B\}_{n}$. By restricting the selection of the fuzzy set $B$ to the  family of fuzzy sets satisfies P4a
	$$\begin{aligned}
		\mathfrak{B}:=\{B\in\mathcal{F}_{c}(\mathbb{R}^{p}): \forall u\in\mathbb{S}^{p-1},\ & \exists C_{u}\subseteq [0,1]\mbox{ with }\nu(C_{u}) = 1 \\ & \mbox{ such that } s_{B}(u,\alpha)\neq 0 \mbox{ } \forall \alpha\in C_{u}
		\},
	\end{aligned}
	$$
	the following result holds for $D_{FS}$ and $D_{mS},$  which is in line with property P4a. 	Property P4b, however, considers a general sequence of fuzzy sets $\{A_{n}\}_{n},$ not allowing for this type of adaptation.
	
	\begin{proposition}\label{proposition1BMSDp}
		For any $\mathcal{X}\in L^{0}[\mathcal{F}_{c}(\mathbb{R}^{p})]$ and $B\in\mathfrak{B},$ we have that 	\begin{itemize}
			\item	$\lim_{n} D_{FS}(A + n\cdot B;\mathcal{X}) = 0,$   with $A\in\mathcal{F}_{c}(\mathbb{R}^{p})$ maximizing $D_{FS}(\cdot;\mathcal{X}).$
			\item  $\lim_{n} D_{mS}(A + n\cdot B;\mathcal{X}) = 0,$  with $A\in\mathcal{F}_{c}(\mathbb{R}^{p})$ maximizing $D_{mS}(\cdot;\mathcal{X}).$
		\end{itemize}
	\end{proposition}
	
	The following result is for $D_{nS}$.
	\begin{theorem}\label{SimplicialSDp}
		For any $\mathcal{X}\in L^{0}[\mathcal{F}_{c}(\mathbb{R}^{p})],$ $D_{nS}(\cdot;\mathcal{X})$   satisfies P1, P4a and P4b 
		for the $d_r$ distances for any $r\in[1,\infty]$ 
		and for the $\rho_r$ distances for any
		$r\in[1,\infty).$
	\end{theorem}
	For property P2, intuitively, the notion of symmetry to be considered would make use of the central symmetry of the support function of a fuzzy set in every $u\in\mathbb{S}^{p-1}$ and $\alpha\in [0,1]$. It is apparent that the relation between this tentative notion of symmetry  and the notion of a fuzzy simplex is $F$-symmetry. As regards properties P3a and P3b, already in the multivariate case the simplicial depth does not
	generally satisfy the analog property M3. Because of these reasons and since naive simplicial fuzzy depth is not one of our recommended fuzzy depth, we do not pursue these properties further.
	
	\section{Empirical simplicial depths}
	\label{datasimulation}
	Given $ \mathcal{H}\subseteq L^{0}[\mathcal{F}_{c}(\mathbb{R}^{p})],$  let   $\mathcal{X}\in\mathcal{H}$ be a fuzzy random variable and
	$\mathcal{X}_{1},\ldots ,\mathcal{X}_{n}$ be independent  random variables  distributed as $\mathcal{X}.$ Let $\mathfrak{X}$ be a fuzzy random variable corresponding to the empirical distribution associated to
	$\mathcal{X}_{1},\ldots ,\mathcal{X}_{n}$. That is, $\mathfrak X$ takes on as values the observed values $\mathcal X_1(\omega),\ldots,\mathcal X_n(\omega)$ (possibly repeated) with probability $n^{-1}$. The simplicial depths associated with this empirical distribution are the empirical or sample simplicial depths.
	
	In Subsection \ref{Emp}, we provide the explicit definitions for the case of $\mathcal{F}_{c}(\mathbb{R})$ in order to illustrate subsequently the behavior of our three proposals. For ease of comparison with Tukey depth, we
	use in Subsection \ref{Real} the same dataset in \cite{primerarticulo}. The behaviour is similar, which is interesting since that distribution is not from $C^{0}[\mathcal{F}_{c}(\mathbb{R}^{p})]$ as assumed by some of our
	theoretical results (Theorem \ref{teoremamSDp}). In order to illustrate the case of fuzzy random variables with continuously distributed support functionals, we generate in Section \ref{Simu} a synthetic sample from a fuzzy
	random variable in $C^{0}[\mathcal{F}_{c}(\mathbb{R}^{p})]$.

	\subsection{Empirical definitions for $\mathcal{F}_{c}(\mathbb{R})$}\label{Emp}

	From  \eqref{ecuacionmSDFp}  and \eqref{ecuacionMSDFp}, $D_{mS}$ and $D_{FS}$ have in common that both involve computing the  function
	$$F_A(u):=\int_{[0,1]} \mathbb{P}(s_{A}(u,\alpha)\in [m_{\mathcal{X}}(u,\alpha),M_{\mathcal{X}}(u,\alpha)]) \dif\nu(\alpha),$$ with $u\in\mathbb{S}^{0}=\{-1,1\}.$
	The difference lies  in the operator over $\mathbb{S}^{0}$ applied to $F_A:$ the average (for $D_{mS}$) and the infimum (for $D_{FS}$). %
	Then, to establish our proposals of emprirical simplicial and modified simplicial fuzzy depth, making use of $\mathcal{X}_{1},\ldots ,\mathcal{X}_{n},$ we calculate $F_A(u)$ for the fuzzy random variable $\mathfrak X$ as
	$\binom{n}{2}^{-1}L^A_n(u)$ with
	\begin{equation}\label{L}
		L^A_n(u):=\sum_{i = 1}^{n}\sum_{j\geq i}^{n}L_{i,j,u}^{A}
	\end{equation}
	and %, for each pair $(\mathcal{X}_{i}, \mathcal{X}_{j})$ with $j\geq i,$ $i,j\in\{1, \ldots, n\}$ and $u\in\mathbb{S}^{0},$
	\begin{equation}\label{ecuacionLiju}
		L_{i,j,u}^{A} := \nu(\{\alpha\in [0,1] : s_{A}(u,\alpha)\in [\min\{s_{\mathcal{X}_{i}}(u,\alpha), s_{\mathcal{X}_{j}}(u,\alpha)\},\max\{s_{\mathcal{X}_{i}}(u,\alpha), s_{\mathcal{X}_{j}}(u,\alpha)\}]  \}).
	\end{equation}
	
	Then, the modified simplicial fuzzy depth based on $\mathcal{J}\subseteq\mathcal{F}_{c}(\mathbb{R})$ of a fuzzy set $A\in\mathcal{J}$
	with respect to
	$\mathfrak{X}$ is
	\begin{equation}\label{definicionDmS}
		D_{mS}(A;\mathfrak{X}) = \int_{\mathbb{S}^{0}}\binom{n}{2}^{-1}L_{n}^{A}(u) \dif\mathcal{V}_{1}(u) = 2^{-1}\binom{n}{2}^{-1}\left[L_{n}^{A}(1) + L_{n}^{A}(-1)\right]
	\end{equation}
	and the  simplicial fuzzy depth based on $\mathcal{J}\subseteq\mathcal{F}_{c}(\mathbb{R})$  of a fuzzy set $A\in\mathcal{J}$
	with respect to  $\mathfrak{X}$ is
	\begin{equation}\label{definicionDFS}
		D_{FS}(A;\mathfrak{X}) = \inf_{u\in\mathbb{S}^{0}}\binom{n}{2}^{-1}L_{n}^{A}(u) 		= \binom{n}{2}^{-1}\min\{L_{n}^{A}(1),L_{n}^{A}(-1)\}.
	\end{equation}
	Similarly, the  naive simplicial fuzzy  depth based on $\mathcal{J}\subseteq\mathcal{F}_{c}(\mathbb{R})$  of a fuzzy set $A\in\mathcal{J}$ with respect to  $\mathfrak{X}$ is
	\begin{equation}\label{ecuacionDnS}
		D_{nS}(A;\mathfrak{X}) = \cfrac{1}{\binom{n}{2}}\sum_{i=1}^{n}\sum_{j\geq i}^{n}I_{i,j}^{A},
	\end{equation}
	where $I_{i,j}^{A}$ equals 1 if  $s_{A}(u,\alpha)\in [\min\{s_{\mathcal{X}_{i}}(u,\alpha),s_{\mathcal{X}_{j}}(u,\alpha)\},\max\{s_{\mathcal{X}_{i}}(u,\alpha),s_{\mathcal{X}_{j}}(u,\alpha\}]$  for every $(u,\alpha)\in\mathbb{S}^{0}\times [0,1]$, and  0 otherwise.
	
	\subsection{Simulated data}\label{Simu}
	
	We draw a sample ($n=100$) from a fuzzy random variable in $C^{0}[\mathcal{F}_{c}(\mathbb{R}^{p})]$. For that, we make use of a random variable whose realizations are trapezoidal fuzzy sets. 
	To construct the fuzzy random variable, we follow the method in \cite{SinovaMedian}. Let  $X_{1}, X_{2}, X_{3}, X_{4}$ be independent and continuous real-valued
	random variables. Let $X_{1}$ be normally distributed with zero mean and standard deviation 10, whereas $X_{2}, X_{3}, X_{4}$ are chi-squared distributions with 1 degree of freedom. Set
	\begin{equation}\label{Xsim}
		\mathcal{X} = \mbox{Tra}(X_{1} - X_{2} - X_{3}, X_{1} - X_{2},X_{1} + X_{2}, X_{1} + X_{2} + X_{4})
	\end{equation}
	which is well-defined since $X_2,X_3,X_4\ge 0$. By construction,
	$$s_{\mathcal X}(-1,\alpha)=-(X_1-X_2-(1-\alpha)X_3)$$
	and
	$$s_{\mathcal X}(1,\alpha)=X_1+X_2+(1-\alpha)X_4,$$
	which are continuous
	variables for each $\alpha\in[0,1]$. Accordingly, $\mathcal{X}\in C^{0}[\mathcal{F}_{c}(\mathbb{R})]$ as required by Theorem \ref{teoremamSDp}.
	
	The choice of the $\chi_1^2$ distribution for $X_3,X_4$ is because it is very skewed (Pearson coefficient: $2\sqrt 2$). That allows us to realize how the depth is affected not just by the location of the core of the trapezoidal
	fuzzy set but also by the slopes of its sides.
	
	To illustrate the performance of the different depth functions, let $\mathcal{X}_{1},\ldots ,\mathcal{X}_{100}$ be independent copies distributed as $\mathcal{X}.$ With some abuse of notation, for $i=1, \ldots, 100,$ each
	$\mathcal{X}_i$ will also denote the observed trapezoidal fuzzy set, represented in each of the plots of Figure \ref{figuracontinuasimplicial}. Thus, we illustrate the performance of each of our three proposals by
	computing,  for $i=1, \ldots, 100,$ each of the depths of $\mathcal{X}_i$ with respect to the corresponding empirical fuzzy random variable $\mathfrak{X}.$ 
	Naive simplicial depth $D_{nS}$ is illustrated in the top row of Figure \ref{figuracontinuasimplicial}, modified simplicial depth $D_{mS}$ in the middle row, and simplicial fuzzy depth $D_{FS}$ in the bottom row. The plots in the first column of Figure
	\ref{figuracontinuasimplicial} represent the five trapezoidal fuzzy values having the largest depth values. These are colored from red (highest depth) to yellow (high depth) and the rest of the
	100 in grey. A zoom of each of these plots highlighting the deepest sets is in the central column of the figure.

	\begin{figure}[htbp]
		\begin{center} 		\includegraphics[width=0.32\linewidth]{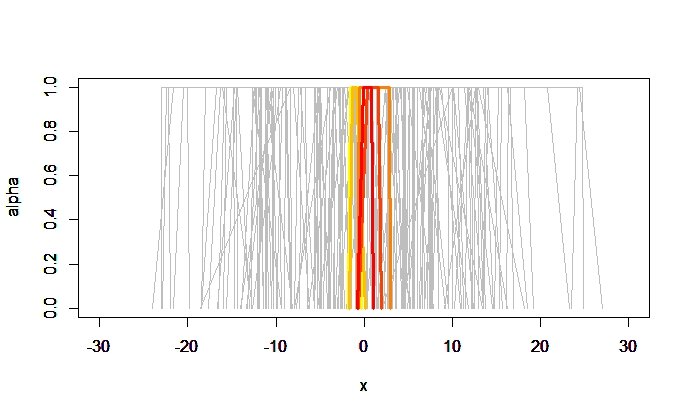} 	
			\includegraphics[width=0.32\linewidth]{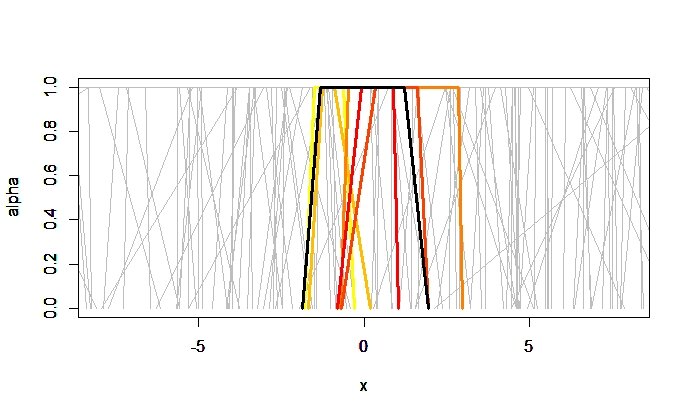} 	
			\includegraphics[width=0.32\linewidth]{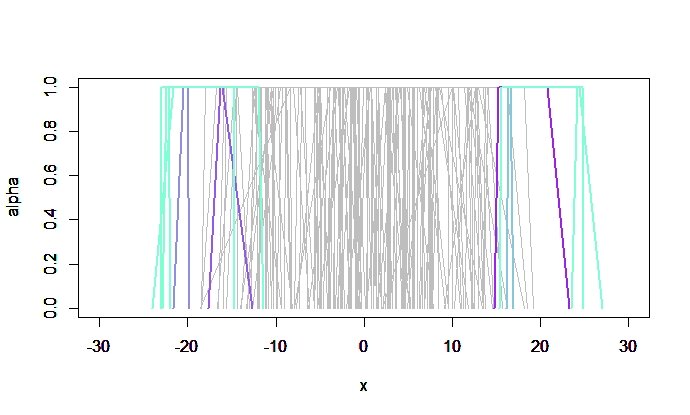} 	
			\includegraphics[width=0.32\linewidth]{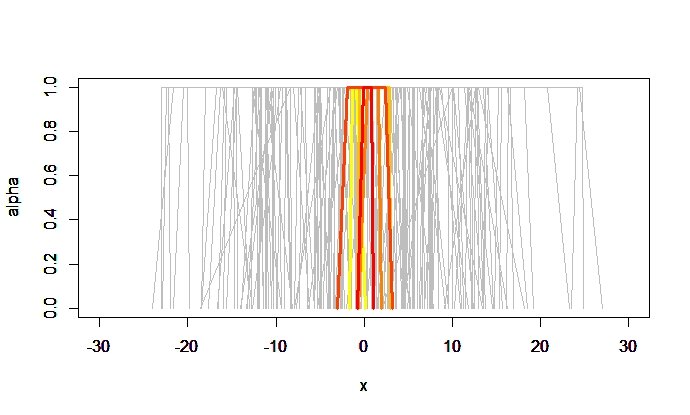} 	
			\includegraphics[width=0.32\linewidth]{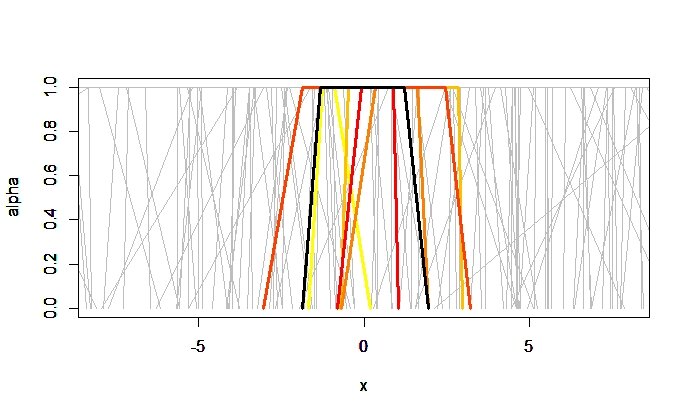} 	    
			\includegraphics[width=0.32\linewidth]{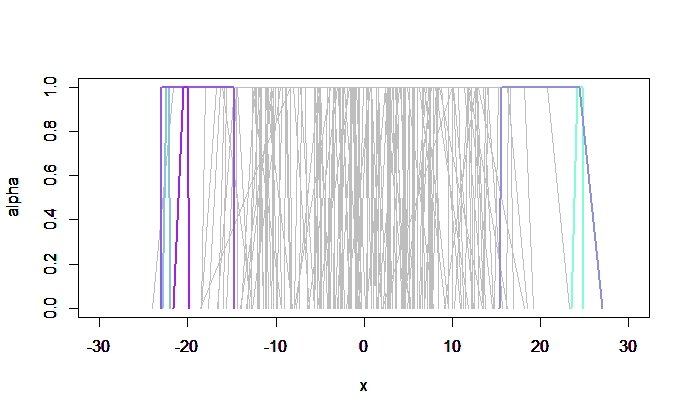} 	
			\includegraphics[width=0.32\linewidth]{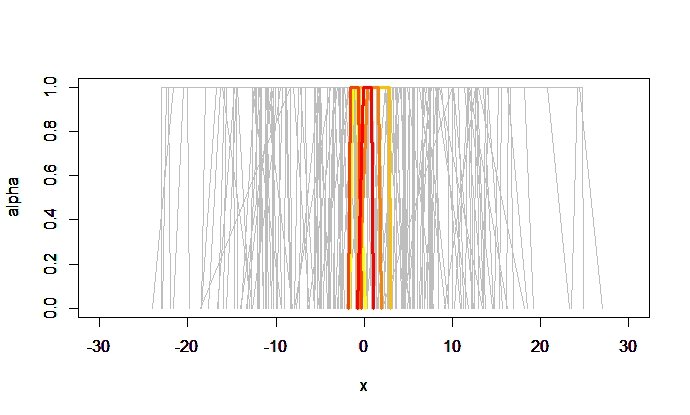} 	 
			\includegraphics[width=0.32\linewidth]{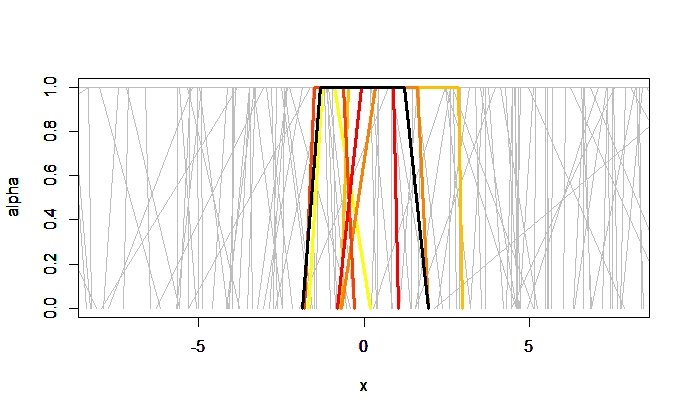} 	 
			\includegraphics[width=0.32\linewidth]{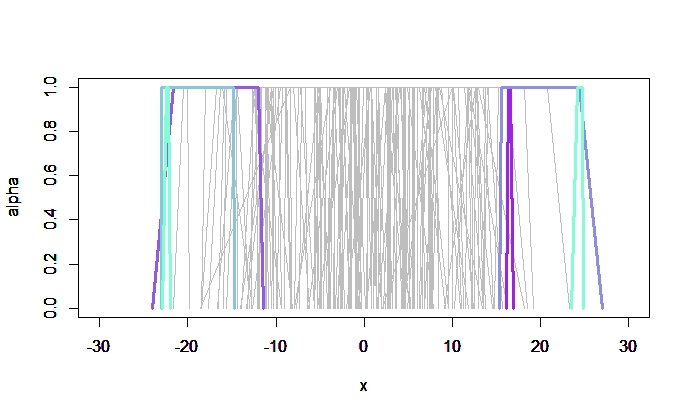}
		\end{center} 	\caption{ 	Illustration of the empirical naive simplicial fuzzy  depth, $D_{nS},$ (top row), the empirical modified simplicial fuzzy depth, $D_{mS},$ (middle row) and the empirical simplicial fuzzy depth,
			$D_{FS,}$ (bottom row) over a sample of trapezoidal fuzzy sets of  size 100 drawn from $\mathcal{X}$ in \eqref{Xsim}. The sample is plotted in grey. The color  in the first and second column plots represent the trapezoidal
			fuzzy sets in the sample corresponding to the 5 larger depth values, with the second column being a  zoom of the first in the interval $[-8,8];$ in order to better observe  the different depth values. Colors range from red
			(highest depth) to yellow (high depth) in the first column. In addition, in the second column the median fuzzy set is highlighted in black. The third column represents the trapezoidal fuzzy sets with the 5 minimal depth
			values for the same depth functions. 		Depth values are shown through the colors, which range from aqua marine blue (lowest depth) to violet (low depth). 		} 	\label{figuracontinuasimplicial}
	\end{figure}

	We also represent, plotted in black in the central column of Figure \ref{figuracontinuasimplicial}, the median fuzzy set, $M$, with respect to the sample $\mathcal{X}_{1},\ldots ,\mathcal{X}_{100}$. Denoting $\mathcal{X}_{i}
	= \text{Tra}(a_{i},b_{i},c_{i},d_{i})$ for every $i\in\{1,\ldots ,100\}$, the median fuzzy set is defined as
	$$M := \text{Tra}(\text{Med}(a_{1}, \ldots, a_{100}),\text{Med}(b_{1}, \ldots, b_{100}),\text{Med}(c_{1}, \ldots,c_{100}),\text{Med}(d_{1}, \ldots, d_{100})).$$
	This coincides with the definition in \cite{SinovaMedian}.
	The median $M$ is not necessarily one of the sample fuzzy sets; and in the particular case of Figure  \ref{figuracontinuasimplicial}, it is not.
	The maximizers of the depth functions $D_{nS}$, $D_{mS}$ and $D_{FS}$ provide alternative definitions of a median fuzzy set. They are in the vicinity of $M$ (represented in yellow in the figure) but they are not identical with $M$.
	
	The right column of Figure \ref{figuracontinuasimplicial} shows the trapezoidal fuzzy sets with the minimal 5 depth values for the three different proposals of simplicial depth. The trapezoidal fuzzy sets with minimal depth are the ones furthest to the left and right, as expected.
	It is observable from the plots that
	the three definitions order the sets with minimal depth in a similar way. The main  difference lies in that $D_{nS}$ gives a high number of ties (observe how many sets are colored in aquamarine blue in the last column of
	the first row). The reason for this is that
	$D_{nS}$   is a sum of indicator functions  \eqref{ecuacionDnS} %, $I_{i,j}^{A}$, which takes value  $1$ when $s_{A}(u,\alpha)$ is contained in the interval $$[\min\{s_{\mathcal{X}_{i}}(u,\alpha),s_{\mathcal{X}_{j}}(u,\alpha)\},\max\{s_{\mathcal{X}_{i}}(u,\alpha),s_{\mathcal{X}_{j}}(u,\alpha\}]$$ for every $(u,\alpha)\in\mathbb{S}^{0}\times [0,1]$.
	while the other two proposals make use the Lebesgue measure [\eqref{ecuacionLiju}, \eqref{definicionDmS} and \eqref{definicionDFS}]. Thus, it is generally more convenient to use the proposals $D_{mS}$ and $D_{FS}$ instead
	of $D_{nS};$ with results for $D_{nS}$ being inappropriate for some applications like classification. The use of a sum of indicator functions versus the Lebesgue measure also explains that
	$D_{nS}$ results in smaller depth values than $D_{mS}$ or $D_{FS}.$

	The main difference between $D_{mS}$ and $D_{FS},$ of a fuzzy set $A\in\mathcal F_c(\R^p),$ is that the first one takes the average  of $L^A_n(u)$ in \eqref{L} between  $u = -1$ and $u = 1$ and the second one its minimum over $u \in\{ 1,
	-1\}$.
	Thus, a fuzzy number $A$ with, for instance, $$L^A_n(1) \mbox{ close to } L^M_n(1) \mbox{ and } L^A_n(-1)\mbox{ far from }L^M_n(-1)$$
	does not take a maximal depth value with $D_{FS}$ but can take it with $D_{mS}.$ This is observed in the central column of Figure \ref{figuracontinuasimplicial}.
	
	A similar phenomenon is observed with the fuzzy numbers taking minimal depth values. The bottom row right column plot in Figure \ref{figuracontinuasimplicial} shows that there exists fuzzy numbers in the sample with minimal depth for
	$D_{FS},$ some are on the left side of the plot and the others on the right side. Among the ones on the left there are those that have, for instance,
	$$L^A_n(-1) \mbox{ far from } L^M_n(-1) \mbox{ while } L^A_n(1) \mbox{ is not as far from } L^{M}_n(1).$$
	Analogously, among the ones on  the right  there those that have, for instance, $$L^A_n(1) \mbox{ far from } L^M_n(1) \mbox{ while } L^A_n(-1)  \mbox{ is not as far from } L^{M}_n(-1).$$  As it observable from the central
	row right column plot in Figure  \ref{figuracontinuasimplicial}, these fuzzy numbers does not necessarily take minimal depth value with $D_{mS},$ as this depth function takes the average between $L^A_n(1)$ and $L^A_n(-1)$.
	
	\subsection{Real data}\label{Real}
	We use the \emph{Trees} dataset (from the {\tt SAFD}  R package for Statistical Analysis of Fuzzy Data), which was first used in \citep{Colubi}. This comes from a reforestation project in the region of Asturias (Northern
	Spain) by the INDUROT forest institute at the University of Oviedo. The project takes into account three species of trees: birch (\textit{Betula celtiberica}), sessile oak (\textit{Quercus petraea}) and rowan
	(\textit{Sorbus aucuparia}).
	
	The most important variable considered is the \textit{quality} of trees, whose observations are trapezoidal fuzzy sets coming from an expert subjective assessment of height, diameter, leaf structure and other features.  The dataset is represented in Figure \ref{figurasimplicial}, where quality is measured in the  x-axis in the range 1--5, from low to perfect quality. The membership values for each trapezoidal fuzzy set are  represented in the y-axis.
	
	The dataset is comprised of 9 different trapezoidal fuzzy values, represented in Figure \ref{figurasimplicial}. Therefore, the assumption in our theoretical study that each support functions has a continuous distribution is violated, which makes it interesting to check the depth functions' behavior.
	From left to right we denote them by $T_{1},\ldots , T_{9}.$ These sets appear in the sample with a certain multiplicity, resulting  in a sample $\mathcal{X}_1, \ldots, \mathcal{X}_n$ of size $n=279.$ Table \ref{tabladifusos} shows the absolute frequency of the fuzzy sets in the sample. We denote by
	$\mathfrak{X}$ the  fuzzy random variable corresponding to the empirical distribution associated to $\mathcal{X}_1, \ldots, \mathcal{X}_n.$ 
	\begin{table}[h]		
		\begin{center}					\begin{tabular}{ccccccccc}							\hline										$T_{1}$ & $T_{2}$ & $T_{3}$ & $T_{4}$ & $T_{5}$ & $T_{6}$ & $T_{7}$ & $T_{8}$ & $T_{9}$ \\ 							 
				\hline 			\\				22 & 16 & 39 & 36 & 85 & 22 & 35 & 12 & 12 \\											\hline						\end{tabular}				\end{center}			\caption{Number of sets
			in the sample for each type of trapezoidal fuzzy set. 	Absolute frequency of each distinct trapezoidal fuzzy set $T_i,$ $i=1, \ldots, 9,$ represented in  Figure \ref{figurasimplicial}.}		\label{tabladifusos}	
	\end{table}
	
	One can observe from Figure \ref{figurasimplicial} that
	\begin{equation}\label{orderSimplicial}
		s_{T_{i}}(1,\alpha)\geq s_{T_{j}}(1,\alpha)\mbox{ and }s_{T_{i}}(-1,\alpha)\leq s_{T_{j}}(-1,\alpha)
	\end{equation}
	for each $\alpha\in [0,1]$ and  $i,j\in \{1,\ldots 9\}$ with $i\leq j$. In fact the inequalities are strict except for the cases of $T_{4}, T_{5}$ and $T_{6}$, where
	\begin{equation}\label{tiesrealdata}
		s_{T_{4}}(-1,0) = s_{T_{5}}(-1,0)\mbox{ and }s_{T_{5}}(1,0) = s_{T_{6}}(1,0).
	\end{equation}
	Taking into account the sample version of $D_{nS}$ in \eqref{ecuacionDnS} and the fact that $I_{i,j}^{A}$ takes  value $1$ if  
	$$
	s_{A}(u,\alpha)\in [\min\{s_{\mathcal{X}_{i}}(u,\alpha),s_{\mathcal{X}_{j}}(u,\alpha)\},\max\{s_{\mathcal{X}_{i}}(u,\alpha),s_{\mathcal{X}_{j}}(u,\alpha)\}]
	$$ 
	for every $(u,\alpha)\in\mathbb{S}^{0}\times [0,1]$ and $0$ otherwise,
	the computation of $D_{nS}(T_{i};\mathfrak{X})$ reduces to computing  the simplicial depth in $\mathbb{R}$ of $s_{T_{i}}(u,\alpha)$ with respect to $s_{\mathfrak{X}}(u,\alpha)$ for some $(u,\alpha)$
	where the inequalities in \eqref{orderSimplicial} are strict. Taking into account \eqref{tiesrealdata}, this is the case of $(u,\alpha) = (1,1),$ for instance.
	Thus
	$$D_{nS}(T_{i};\mathfrak{X}) = SD(s_{T_{i}}(1,1);s_{\mathfrak{X}}(1,1))$$ for each $i\in\{1,\ldots ,9\}$.
	
	Taking into account the order given by \eqref{orderSimplicial} of $\{T_{i}\}_{i=1}^{9}$  for each $(u,\alpha)\in \mathbb{S}^{0}\times [0,1]$ and  $i\leq j$ with $i,j\in\{1,\ldots 9\},$
	we have that $L_{i,j,u}^{T_{k}} = 1$ for each $k\in[i, j]$ and $0$ otherwise. Considering the sample versions of $D_{mS}$ and $D_{FS}$ in  \eqref{definicionDmS} and \eqref{definicionDFS}, we have that in this case
	the three depth proposals coincide, that is,
	$$D_{nS}(T_{i};\mathfrak{X}) = D_{mS}(T_{i};\mathfrak{X}) = D_{FS}(T_{i};\mathfrak{X})$$ for each $i\in\{1,\ldots ,9\}$. Thus,  in computing the depth of an element in the dataset
	with respect to the empirical fuzzy random variable, we obtain the same depth value
	independently of which of the three simplicial based fuzzy depths is used. The left plot of  Figure \ref{figurasimplicial} represents in the color the depth values of each of the 9 distinct trapezoidal elements in the
	dataset. Colors range from brown (high depth) to yellow (low depth).
	
	\begin{figure}[htbp]
		\begin{center} 		\includegraphics[width=0.49\linewidth]{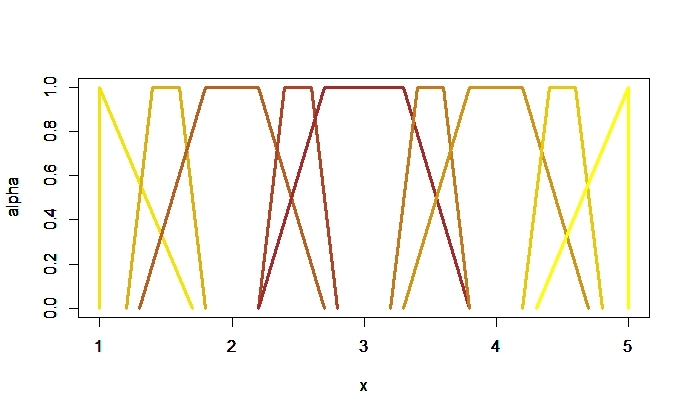} 	
			\includegraphics[width=0.49\linewidth]{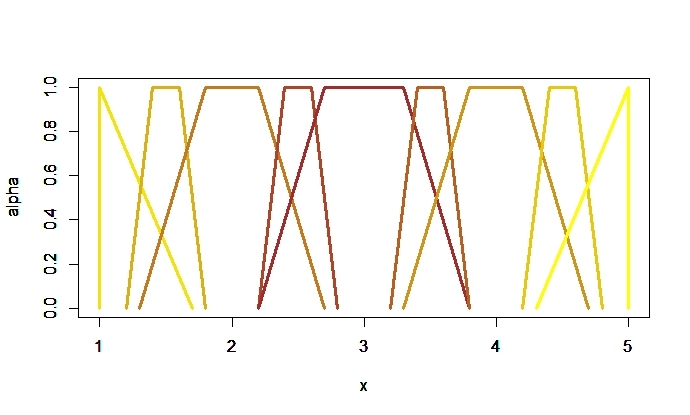} 	
		\end{center} 	\caption{ 		Display of the fuzzy sets in the \emph{Trees} dataset. 		In the first column, color is assigned based on the Simplicial depth of each fuzzy set in the empirical distribution. Colors range from brown (high depth) to yellow (low depth). 		The second column applies the same procedure but considering the Tukey depth.} 	\label{figurasimplicial}\end{figure}
	
	From Figure \ref{figurasimplicial} we can observe that the order induced in the dataset by  the simplicial and Tukey fuzzy depth functions is similar.  In fact, the only difference is $T_{3}$ and $T_{6}.$ In the case of the
	simplicial fuzzy depths, we have that $T_{3}$ is the third deepest set and $T_{6}$ is the fourth, while we obtain the reverse using the Tukey fuzzy depth. Let us explain where this difference comes from. If we observe Table
	\ref{tabladifusos} we have that $T_{3}$ has 39 repetitions in the sample while $T_{6}$ only has 22. On the other hand the diameter of 0-level and the 1-level of $T_{3}$ is greater than the diameter of the 0-level and the
	1-level of $T_{6}$. Thus, taking into account that the weight of $T_{3}$ in the sample is greater than the weight of $T_{6},$ as the Tukey depth is defined as a minimum, it could be an explanation of why
	$$D_{FT}(T_{6};\mathcal{X}) > D_{FT}(T_{3};\mathcal{X}).$$ 
	Meanwhile, as the simplicial depth for fuzzy sets is defined by an integral, it takes more into account the weights of the different sets of the sample and
	depreciates what happens in one single point.
	
	\section{Discussion}
	\label{discussion}
	
	Simplicial depth is one of the most widely used depth functions in  multivariate statistics. It is built over the notion of simplex in $\mathbb{R}^{p}$. In the space of fuzzy sets, the notion of simplex is not an obvious one. With the characterization introduced in Proposition \ref{proposicionsimplicieproyeccion} of simplices in the multivariate space, we justify the notion of simplex in $\mathcal{K}_{c}(\mathbb{R}^{p})$ and
	extend it to the fuzzy setting, working $\alpha$-level by $\alpha$-level (Definition \ref{pseudosimplex}).
	Making use of this notion, we propose a straightforward definition  of simplicial depth for the fuzzy setting and two  elaborate and sounded definitions.
	\begin{itemize}
		\item The naive simplicial fuzzy depth \eqref{DnS}, $D_{nS}$, is equivalent to the  multivariate simplicial depth. We prove some properties for it in Theorem \ref{SimplicialSDp} and  show it may result in a high number
		of ties at zero, which  is not desirable for instance in classification problems.
		\item The modified simplicial fuzzy depth (Definition \ref{DmS}),  $D_{mS},$  improves the  naive simplicial fuzzy depth analogously to how the modified band depth improves the band depth; resulting in less zero depth
		values.
		\item The simplicial fuzzy depth (Definition \ref{DFS}), $D_{FS}$, transforms the modified simplicial fuzzy depth in the direction %of the style
		of the Tukey depth; doing so by applying the  infimum over $\mathbb{S}^{p-1}$ instead of the expected value.
	\end{itemize}

	Although it is clear throughout the paper the authoritativeness  of $D_{mS}$ and  $D_{FS}$ over $D_{nS},$ there is not a clear winner between $D_{mS}$ and  $D_{FS}.$ The practical similarities and differences between them are
	discussed in Example \ref{Exfig} and Subsection \ref{Simu}. 
	Their properties are collected in Theorem \ref{teoremamSDp} and Proposition \ref{proposition1BMSDp}. For some of these properties it is required fuzzy random variables to satisfy certain type of continuity.
	This is inherited from the fact that the multivariate simplicial depth requires of continuous distributions to satisfy the notion of multivariate depth.

	Our three proposals neither satisfy the notion of semilinear nor of geometric depth function in \citep{primerarticulo} because of the lack of satisfaction of the entirety of the properties constituting these notions (Section
	\ref{properties}).
	However, as we can see in the illustrations
	in Section \ref{datasimulation}, the behavior of the three proposals is similar in practice. As shown there, it is also similar to that of the Tukey fuzzy depth, despite Tukey does satisfy both notions and the comparison
	is done with respect to a  fuzzy random variable that does not satisfy the  continuity properties required in Theorem \ref{teoremamSDp}.

	For future work, it is interesting to study more instances of fuzzy depth, creating a library of depth functions for the fuzzy setting. Also, we consider it is compelling to study more properties for the Tukey fuzzy depth and the
	simplicial fuzzy depths, such as convergence of the sample depth to the population depth (consistency) and their continuity or semicontinuity properties.
	
		\section{Proofs}\label{proofs}
	
	\begin{proof}[Proof of Proposition \ref{proposicionsimplicieproyeccion}]
		Let us denote 	$$\mathcal{C} := \{x\in\mathbb{R}^{p} : \langle u,x\rangle\in [m(u), M(u)] \text{ for all } u\in\mathbb{S}^{p-1}\}.$$
		
		First, we prove $S[x_{1},\ldots,x_{p+1}]\subseteq\mathcal{C}$. Let $x\in S[x_{1},\ldots,x_{p+1}]$. By \eqref{S}, there exists $\lambda_{1},\ldots,\lambda_{p+1}\geq 0$ with $\sum_{i = 1}^{p+1}\lambda_{i} = 1$ such that $x
		= \sum_{i = 1}^{p+1}\lambda_{i}x_{i}.$ For any fixed  direction $u\in\mathbb{S}^{p-1}$, we have 	$	\langle u,x\rangle = \sum_{i = 1}^{p+1}\lambda_{i}\langle u,x_{i}\rangle. 	$
		As $\lambda_{i}\in [0,1]$ for all $i= 1, \ldots, p+1,$ we have that $\langle u,x\rangle\in [m(u), M(u)];$ and, consequently,  $x\in \mathcal C$.
		
		Now, let $x\in \mathcal{C}$ and suppose for a contradiction that $x\not\in S[x_{1},\ldots,x_{p+1}]$. The simplex $S[x_{1},\ldots,x_{p+1}]$ and the set $\{x\}$ are closed, convex and bounded subsets of $\mathbb{R}^{p}$.
		By the Hyperplane Separation Theorem (see, e.g., \citep{separatinghyperplane}),  there exist $u\in\mathbb{R}^{p}$ and $b\in\mathbb{R}$ such that  $\langle u,x\rangle > b$ and  $\langle u,s\rangle < b$ for all $s\in
		S[x_{1},\ldots,x_{p+1}]$. This implies that $\langle u,x\rangle > \langle u,s\rangle$ for all $s\in S[x_{1},\ldots,x_{p+1}]$. Normalizing the vector $u$, $\bar{u}\in S^{p-1}$, we have that $\langle \bar{u},x\rangle >
		M(\bar{u}).$ This is a contradiction with the fact that $x\in \mathcal{C}$. Thus, $x\in S[x_{1},\ldots,x_{p+1}].$
	\end{proof}
	
	\begin{proof}[Proof of Proposition \ref{PropositionSimplices}]
		Let $A_{1},\ldots ,A_{p+1}\in\mathcal{K}_{c}(\mathbb{R}^{p})$ and $A\in\mathcal{K}_{c}(\mathbb{R}^{p})$ such that there exist real numbers $\lambda_{1},\ldots ,\lambda_{p+1}\geq 0$ with $\sum_{i=1}^{p+1}\lambda_{i} = 1$
		and $A = \sum_{i=1}^{p+1}\lambda_{i}\cdot A_{i}$. 	By \eqref{soportesuma}, 	$ 	s_{A}(u) = \sum_{i=1}^{p+1}\lambda_{i}\cdot s_{A_{i}}(u) 	$ 	for every $u\in\mathbb{S}^{p-1}$. Thus, for every $u\in\mathbb{S}^{p-1},$ 	
		$$ 	m(u) = (\sum_{i=1}^{p+1}\lambda_{i})m(u)\leq s_{A}(u)\leq (\sum_{i=1}^{p+1}\lambda_{i})M(u) = M(u). 	$$ 	 Then $A\in S_c[A_{1},\ldots , A_{p+1}]$.
	\end{proof}
	
	\begin{proof}[Proof of Proposition \ref{aaa}]
		For any $A\in\pfc(\R)$, since $\mathbb S^0=\{-1,1\}$ we have $$s_A(1,\alpha)=\sup A_\alpha, \quad s_A(-1,\alpha)=\sup\{-x\mid x\in A_\alpha\}=-\inf A_\alpha.$$ For any fixed $\alpha$, inequality $m(u\alpha)\le
		s_A(u,\alpha)\le M(u\alpha)$ will hold for $u=1$ if and only if $$\min\{\sup (A_1)_\alpha,\sup (A_2)_\alpha\}\le \sup A_\alpha \le \max\{\sup (A_1)_\alpha,\sup (A_2)_\alpha\}$$ which, taking into account the assumption
		$A_1\preceq A_2$, is equivalent to $$\sup (A_1)_\alpha\le \sup A_\alpha\le \sup (A_2)_\alpha.$$ In its turn, the inequality will hold for $u=-1$ if and only if $$\min\{-\inf (A_1)_\alpha,-\inf (A_2)_\alpha\}\le -\inf
		A_\alpha \le \max\{-\inf (A_1)_\alpha,-\inf (A_2)_\alpha\}$$ or, multiplying all terms by $-1$, $$\max\{\inf (A_1)_\alpha,\inf (A_2)_\alpha\}\ge \inf A_\alpha \ge \min\{\inf (A_1)_\alpha,\inf (A_2)_\alpha\}$$ which, again by
		the assumption $A_1\preceq A_2$, is the same thing as $$\inf (A_2)_\alpha \ge \inf A_\alpha \ge \inf (A_1)_\alpha.$$ The conjunction of those two conditions is just $A_1\preceq A\preceq A_2$. Hence
		$$S_F[A_1,A_2]=\{A\in\pfc(\R^d): A_1\preceq A\preceq A_2\}.$$
	\end{proof}
	
	\begin{proof}[Proof that the naive simplicial fuzzy depth is well defined]
		We need to show that the event $$\{s_A(u,\alpha)\in[m_{\mathcal{X}}(u,\alpha),M_{\mathcal{X}}(u,\alpha)]\mbox{ for all }(u,\alpha)\in\mathbb{S}^{p-1}\times[0,1]\}$$ $$=\bigcap_{u\in
			\mathbb{S}^{p-1}}\bigcap_{\alpha\in[0,1]}\{s_A(u,\alpha)\in[m_{\mathcal{X}}(u,\alpha),M_{\mathcal{X}}(u,\alpha)]\}$$ is measurable.
		
		First, for each fixed $u,\alpha$ and $i=1,\ldots,p-1$, the mapping $s_{\mathcal X_i}(u,\alpha)$ is a random variable  \cite[Lemma 4]{Kra}. Subsequently,
		$$\Omega_{u,\alpha}:=\{s_A(u,\alpha)\in[m_{\mathcal{X}}(u,\alpha),M_{\mathcal{X}}(u,\alpha)]\}$$ $$=\left(\bigcup_{i=1}^{p+1}\{s_{\mathcal X_i}(u,\alpha)\le s_A(u,\alpha)\}\right)\cap\left(\bigcup_{i=1}^{p+1}\{s_{\mathcal
			X_i}(u,\alpha)\ge s_A(u,\alpha)\}\right)$$ is measurable.
		
		Taking  $D$  a countable dense subset of $[0,1]$ such that $0\in D,$ let us prove
		\begin{equation}\label{fgh}
			\bigcap_{\alpha\in[0,1]}\Omega_{u,\alpha}=\bigcap_{\alpha\in D}\Omega_{u,\alpha}\quad \hbox{ for each fixed }u\in\mathbb S^{p-1}.
		\end{equation}
		The left-to-right inclusion is trivial. For the converse inclusion, assume for now that $\alpha\in(0,1]$. We can construct a sequence of elements of $D$ converging to $\alpha$ from the left (which is why $\alpha>0$ is
		needed). Indeed, for each $n\in\N$ with $n>\alpha^{-1}$ consider the open interval $(\alpha-n^{-1},\alpha)$. It contains some $\alpha_n\in D,$ because of $D$ being dense.
		Since $\alpha-n^{-1}<\alpha_n<\alpha$, we have $\alpha_n\to\alpha^-$. Now the mapping $s_A(u,\cdot)$ is left continuous \citep{MingSupport}. Similarly, for any arbitrary $\omega\in\Omega$, the $s_{\mathcal
			X_i(\omega)}(u,\cdot)$ are left continuous, whence $m(u,\cdot)$ and $M(u,\cdot)$ are too. For any $\omega\in\bigcap_{\alpha\in D}\Omega_{u,\alpha}$ we have
		$$m(u,\alpha_n)\le s_A(u,\alpha_n)\le M(u,\alpha_n)$$ (please note the unspecified dependence of $m$ and $M$ on $\omega$ via the $s_{\mathcal X_i}$). By the left continuity, also $$m_{\mathcal{X}}(u,\alpha)\le
		s_A(u,\alpha)\le M_{\mathcal{X}}(u,\alpha).$$ This means that $\omega$ is in $\Omega_{u,\alpha}$ for each $\alpha\in(0,1]$. The case $\alpha=0$ holds as well since we chose $D$ with $0\in D$. Accordingly, (\ref{fgh}) holds.
		That proves that each $\bigcap_{\alpha\in[0,1]}\Omega_{u,\alpha}$, being a countable intersection of measurable events, is measurable.
		
		$\mathbb S^{p-1}$, being a compact metric space, is separable. Let us take a countable dense subset $D'\subseteq \mathbb S^{p-1}$. The proof will be complete if we show $$\bigcap_{u\in
			\mathbb{S}^{p-1}}\bigcap_{\alpha\in[0,1]}\Omega_{u,\alpha}=\bigcap_{u\in D'}\bigcap_{\alpha\in[0,1]}\Omega_{u,\alpha},$$ since the left-hand side is the event we wish to prove measurable and the right-hand side is a
		countable intersection of measurable events. As before, only the right-to-left inclusion need be proved. Let us fix an arbitrary $u^*\in\mathbb S^{p-1}$. Due to $D'$ being dense, there exists a sequence $u_n\to u^*$ with
		$u_n\in D'$. Whenever $\omega\in \bigcap_{u\in D'}\bigcap_{\alpha\in[0,1]}\Omega_{u,\alpha}$, we have $$m(u_n,\alpha)\le s_A(u_n,\alpha)\le M(u_n,\alpha) \hbox{ for all }\alpha\in[0,1].$$ By the continuity of the support
		functions for fixed $\alpha$ \citep{MingSupport}, the convergence $u_n\to u^*$ implies $$m(u^*,\alpha)\le s_A(u^*,\alpha)\le M(u^*,\alpha) \hbox{ for all }\alpha\in[0,1].$$ That establishes $$\bigcap_{u\in
			D'}\bigcap_{\alpha\in[0,1]}\Omega_{u,\alpha}\subseteq \bigcap_{\alpha\in[0,1]}\Omega_{u^*,\alpha}.$$ By the arbitrariness of $u^*$, $$\bigcap_{u\in D'}\bigcap_{\alpha\in[0,1]}\Omega_{u,\alpha}\subseteq \bigcap_{u\in
			\mathbb{S}^{p-1}}\bigcap_{\alpha\in[0,1]}\Omega_{u,\alpha},$$ as wished. The proof is complete.
	\end{proof}
	
	\begin{proof}[Proof that the modified simplicial fuzzy depth is well defined]
		In order to show that both expressions defining $D_{mS}$ make sense and are equal, and justify the claim that Fubini's Theorem applies, we start by considering the following subset of the product measurable space $\Omega
		\times\mathbb S^{p-1}\times[0,1]$: $$Z:=\{(\omega,u,\alpha)\in \Omega\times\mathbb S^{p-1}\times[0,1]: \min_{1\le i\le p+1}s_{\mathcal X_i(\omega)}(u,\alpha)
		\le s_A(u,\alpha)\le \max_{1\le i\le p+1}s_{\mathcal X_i(\omega)}(u,\alpha)\}.$$
		Let us prove that $Z$ is measurable, i.e., it is in the product $\sigma$-algebra of $\Omega\times\mathbb S^{p-1}\times[0,1]$. Bear in mind that $Z$ is {\em not} the event
		$\bigcap_u\bigcap_\alpha\Omega_{u,\alpha}\subseteq\Omega$ from the previous proof.
		
		Given any fuzzy random variable $\mathcal X$, the support mapping $$\tilde s: (\omega,u,\alpha)\in \Omega\times\mathbb S^{p-1}\times[0,1]\mapsto s_{\mathcal X(\omega)}(u,\alpha)\in\R$$ is a random variable, by \cite[Lemma
		4]{Kra} or \cite[Proposition 4.6]{Joint}. Denote by $\tilde s_{\mathcal X_i}$ the support mapping of each $\mathcal X_i$. Also consider the support mapping $\tilde s_A$ of $A$ seen as a degenerate fuzzy random variable,
		namely $\tilde s_A(\omega,u,\alpha)=s_A(u,\alpha)$. Then $$Z=\left(\bigcup_{i=1}^{p+1}\{\tilde s_{\mathcal X_i}\le \tilde s_A\}\right)\cap\left(\bigcup_{i=1}^{p+1}\{\tilde s_{\mathcal X_i}\ge \tilde s_A\}\right),$$ which is
		a measurable event since the $\tilde s_{\mathcal X_i}$ and $\tilde s_A$ are all random variables. And, accordingly, its indicator function $I_Z:\Omega \times\mathbb S^{p-1}\times[0,1]\to \{0,1\}$ is measurable (and
		integrable against probability measures, since it is bounded).
		
		By the Fubini's Theorem, $$\int_{\Omega\times\mathbb S^{p-1}\times[0,1]}I_Z \dif(\mathbb P\otimes\mathcal V_p\otimes\nu) =\int_\Omega \int_{\mathbb S^{p-1}\times[0,1]}I_Z(\omega,u,\alpha) \dif(\mathcal
		V_p\otimes\nu)(u,\alpha)\dif\mathbb P(\omega)$$ $$=\int_{\mathbb S^{p-1}\times[0,1]}\int_\Omega I_Z(\omega,u,\alpha) \dif\mathbb P(\omega) \dif(\mathcal V_p\otimes\nu)(u,\alpha).$$ Now, for each $\omega\in\Omega$,
		$$\int_{\mathbb S^{p-1}\times[0,1]}I_Z(\omega,u,\alpha) \dif(\mathcal V_p\otimes\nu)(u,\alpha) =(\mathcal V_p\otimes\nu)(\{(u,\alpha)\mid I_Z(\omega,u,\alpha)=1\})$$ $$=(\mathcal V_p\otimes\nu)(\{(u,\alpha)\mid
		(\omega,u,\alpha)\in Z\}) =(\mathcal V_p\otimes\nu)(\{(u,\alpha)\mid m_{\mathcal{X}}(u,\alpha)\le s_A(u,\alpha)\le M_{\mathcal{X}}(u,\alpha)\})$$ whence the second term in the chain of identities is $$\int_\Omega
		\int_{\mathbb S^{p-1}\times[0,1]}I_Z(\omega,u,\alpha) \dif(\mathcal V_p\otimes\nu)(u,\alpha)\dif\mathbb P(\omega)$$ $$=E\left[(\mathcal V_p\otimes\nu)(\{(u,\alpha)\mid m_{\mathcal{X}}(u,\alpha)\le s_A(u,\alpha)\le
		M_{\mathcal{X}}(u,\alpha)\})\right].$$ Moreover, for each $(u,\alpha)$, $$\int_\Omega I_Z(\omega,u,\alpha) \dif\mathbb P(\omega)=\mathbb P(\{\omega\in\Omega\mid m_{\mathcal{X}}(u,\alpha)\le s_A(u,\alpha)\le
		M_{\mathcal{X}}(u,\alpha)\})$$ whence the third term in the chain of identities is, applying again the Fubini's Theorem
		\begin{equation}\label{ecuacionsimplicial1}
			\begin{aligned} &\int_{\mathbb S^{p-1}\times[0,1]}\int_\Omega I_Z(\omega,u,\alpha) \dif\mathbb P(\omega) \dif(\mathcal V_p\otimes\nu)(u,\alpha)\\ &=\int_{\mathbb S^{p-1}}\int_{[0,1]}\mathbb P(\{\omega\in\Omega\mid
				m_{\mathcal{X}}(u,\alpha)\le s_A(u,\alpha)\le M_{\mathcal{X}}(u,\alpha)\})\dif\nu(\alpha) \dif\mathcal{V}_{p}(u).
			\end{aligned}
		\end{equation}
		Those are the expressions for $D_{mS}(A;\mathcal X)$ in  \eqref{ecuacionmSDFp} and \eqref{ecuacionsimplicial1}, which are therefore well defined and indeed equivalent since both equal $\int_{\Omega\times\mathbb S^{p-1}
			\times[0,1]}I_Z \dif(\mathbb P\otimes\mathcal V_p\otimes\nu) $.
	\end{proof}
	
	\begin{proof}[Proof that the simplicial fuzzy depth is well defined]
		It is similar to the proof for the modified simplicial fuzzy depth, by fixing each individual $u\in \mathbb S^{p-1}$ and considering the measurable mapping $I_Z(\cdot,u,\cdot)$.
	\end{proof}
	
	\begin{proof}[Proof of Proposition \ref{Tecuacion2MSDp}]
		Define the events  $Q := \{m_{\mathcal{X}}(u,\alpha)\leq s_{U}(u,\alpha)\}$ and $R: = \{M_{\mathcal{X}}(u,\alpha)\geq s_{U}(u,\alpha)\}.$ Taking into account $$\mathbb{P}(Q^{c}\cap R^{c})\leq
		\mathbb{P}(m_{\mathcal{X}}(u,\alpha)> M_{\mathcal{X}}(u,\alpha))=0,$$ we obtain $$\mathbb{P}(Q\cap R) = 1- \mathbb P(Q^c\cup R^c)= 1 - \mathbb{P}(Q^{c}) - \mathbb{P}(R^{c}).$$ 	Besides, as
		$\mathcal{X}_{1},\ldots,\mathcal{X}_{p+1}$ are independent observations of $\mathcal{X}$, we have that $$s_{\mathcal{X}_{1}}(u_{0},\alpha_{0}),\ldots,s_{\mathcal{X}_{p+1}}(u_{0},\alpha_{0})$$ are independent random
		variables. Then $$\mathbb{P}(Q^{c})=\mathbb{P}(s_{\mathcal{X}_{1}}(u,\alpha) > s_{U}(u,\alpha))^{p+1}
		\mbox{ and }
		\mathbb{P}(R^{c})=\mathbb{P}(s_{\mathcal{X}_{1}}(u,\alpha) < s_{U}(u,\alpha))^{p+1}.$$
		All this together provides the result. In the particular case that $\mathcal{X}\in C^{0}[\mathcal{F}_{c}(\mathbb{R}^{p})],$	the random variable $s_{\mathcal{X}}(u,\alpha)$ is continuous, therefore
		$\mathbb{P}(s_{\mathcal{X}_{1}}(u,\alpha) = s_{U}(u,\alpha))=0.$
	\end{proof}
	
	\begin{proof}[Proof of Theorem \ref{teoremamSDp}]
		\ \\
		\emph{Property P1 for $D_{mS}$ and   $D_{FS}.$} 	Let $M\in\mathcal{M}_{p\times p}(\mathbb{R})$ be a non-singular matrix and $A,B\in\mathcal{F}_{c}(\mathbb{R}^{p}).$ Let us consider independent observations
		$\mathcal{X}_{1},\ldots,\mathcal{X}_{p+1}$ of $\mathcal{X}$ and denote, for any $u\in\mathbb{S}^{p-1}$ and $\alpha\in [0,1],$ 	\begin{align*}\bar{m}_{\mathcal{X}}(u,\alpha) := \min\{s_{M\cdot \mathcal{X}_{1} +
				B}(u,\alpha),\ldots,s_{M\cdot \mathcal{X}_{p+1} + B}(u,\alpha)\}\\ 		\bar{M}_{\mathcal{X}}(u,\alpha) := \max\{s_{M\cdot \mathcal{X}_{1} + B}(u,\alpha),\ldots,s_{M\cdot \mathcal{X}_{p+1} + B}(u,\alpha)\} 		
			.\end{align*} 	 From the properties of the minimum and  maximum, and  \eqref{soportesuma},
		\begin{align*}\bar{m}_{\mathcal{X}}(u,\alpha) = \min\{s_{M\cdot \mathcal{X}_{1}}(u,\alpha),\ldots,s_{M\cdot \mathcal{X}_{p+1}}(u,\alpha)\} + s_{B}(u,\alpha)\\
			\bar{M}_{\mathcal{X}}(u,\alpha) = \max\{s_{M\cdot \mathcal{X}_{1}}(u,\alpha),\ldots,s_{M\cdot \mathcal{X}_{p+1}}(u,\alpha)\} + s_{B}(u,\alpha)
			.\end{align*} 	 Making use of the function 	 \begin{eqnarray*}\label{g} 	 g : \mathbb{S}^{p-1}\rightarrow\mathbb{S}^{p-1} \mbox{ with } g(u) = (1/\|M^{T}u\|)M^{T}u 	 \end{eqnarray*} 	  and \eqref{pa}, we obtain 	
		\begin{align*} 	\bar{m}_{\mathcal{X}}(u,\alpha) = \|M^{T}\cdot u\|\cdot\min\{s_{\mathcal{X}_{1}}(g(u),\alpha),\ldots,s_{\mathcal{X}_{p+1}}(g(u),\alpha)\} + s_{B}(u,\alpha) 	\\ 	\bar{M}_{\mathcal{X}}(u,\alpha) =
			\|M^{T}\cdot u\|\cdot\max\{s_{\mathcal{X}_{1}}(g(u),\alpha),\ldots,s_{\mathcal{X}_{p+1}}(g(u),\alpha)\} + s_{B}(u,\alpha) 	.\end{align*} Similarly, $s_{M\cdot A+B}(u,\alpha)=\|M^T\cdot u\|\cdot s_A(g(u),\alpha)$.
		Consequently, as $g$ is a bijective map, 	\begin{equation} 		\begin{aligned}\nonumber 			&\{(u,\alpha)\in\mathbb{S}^{p-1}\times [0,1] : s_{A}(u,\alpha)\in [m_{\mathcal{X}}(u,\alpha),M_{\mathcal{X}}(u,\alpha)]
				\} = \\ \nonumber 			&\{(u,\alpha)\in\mathbb{S}^{p-1}\times [0,1] : s_{M\cdot A+B}(u,\alpha)\in [\bar{m}_{\mathcal{X}}(u,\alpha),\bar{M}_{\mathcal{X}}(u,\alpha)]\}. 		\end{aligned} 	\end{equation} 	Thus,
		$D_{mS}(A;\mathcal{X}) = D_{mS}(M\cdot A + B;M\cdot\mathcal{X} + B).$
		
		The proof for $D_{FS}$ is analogous. 	\vspace{.4cm}
		
		\emph{Property P2 for $D_{mS}$ and $D_{FS}.$} 	Let $\mathcal{X}\in C^{0}[\mathcal{F}_{c}(\mathbb{R}^{p})]$ be $F$-symmetric with respect to some fuzzy set  $A\in\mathcal{F}_{c}(\mathbb{R}^{p}).$ 		We begin by
		maximizing the integrand in \eqref{ecuacionmSDFp}, which, by Proposition \ref{Tecuacion2MSDp} for $\mathcal{X}\in C^{0}[\mathcal{F}_{c}(\mathbb{R}^{p})],$ is $ 1 - [1 - F_{u,\alpha}(s_{U}(u,\alpha))]^{p+1} -
		[F_{u,\alpha}(s_{U}(u,\alpha))]^{p+1}.$ This is equivalent to minimizing 	\begin{equation}\label{expresionMSDpminimo} 		[1 - F_{u,\alpha}(s_{U}(u,\alpha))]^{p+1} + [F_{u,\alpha}(s_{U}(u,\alpha)]^{p+1}. 	
		\end{equation}
		
		Considering the function 	\begin{eqnarray}\label{f}f : [0,1]\rightarrow\mathbb{R} \mbox{ with  }f(x) = (1-x)^{p+1} + x^{p+1}, 	\end{eqnarray} with derivative $f'(x)=(p+1)(x^p-(1-x)^p)$, the expression in
		\eqref{expresionMSDpminimo} is the composition of $F_{u,\alpha}$ and $f$. 		The function $F_{u,\alpha}$ is non-decreasing and $f$ is strictly decreasing in $[0,1/2]$ and strictly increasing in $[1/2,1],$ with  a minimum
		at $1/2.$ Thus   \eqref{expresionMSDpminimo} is minimized at any $t\in\mathbb{R}$ such that $F_{u,\alpha}(t) = 1/2$ for all $u\in\mathbb{S}^{p-1}$ and $\alpha\in [0,1]$.	
		By \eqref{Amedian} and the assumption that $\mathcal X\in C^0[\pfc(\R^d)]$, it follows that $s_{A}(u,\alpha)$ is one such $t$ 	for each $u\in\mathbb{S}^{p-1}$ and $\alpha\in [0,1]$.
		
		Since $A$ maximizes the integrand in \eqref{ecuacionmSDFp} and \eqref{ecuacionMSDFp} for each $(u,\alpha)$, clearly $A$ maximizes both $D_{mS}(\cdot,\mathcal X)$ and $D_{FS}(\cdot,\mathcal X)$.
		\vspace{.4cm}

		\emph{Property P3a  for $D_{mS}.$} 	Let $B\in\mathcal{F}_{c}(\mathbb{R}^{p})$ and $\lambda\in [0,1]$. 	It suffices to prove that $D_{mS}((1-\lambda)A + \lambda B;\mathcal{X}) - D_{mS}(B;\mathcal{X})\geq 0$. 	Recall
		that $\mathcal{X}\in C^0[\pfc(\R^d)]$ is $F$-symmetric with respect to $A$. Thus each $s_{\mathcal{X}}(u,\alpha)$ is a continuous random variable which is centrally symmetric with respect to $s_{A}(u,\alpha)$ and
		$F_{u,\alpha}(s_{A}(u,\alpha)) = 1/2$. %for every $u\in\mathbb{S}^{p-1}$ and $\alpha\in[0,1]$. 	
		Set 	\begin{equation}\label{x} 	x_{u,\alpha}^{\lambda} : = (1-\lambda)s_{A}(u,\alpha) + \lambda s_{B}(u,\alpha). 	
		\end{equation} 	By \eqref{ecuacionmSDFp}, Proposition \ref{Tecuacion2MSDp} and the linearity of the support function, 	\begin{equation}\label{ecuacionSimplicial} 		\begin{aligned} 			D_{mS}((1-\lambda)\cdot A +
				\lambda\cdot B;\mathcal{X}) - D_{mS}(B;\mathcal{X}) = \int_{\mathbb{S}^{p-1}}\int_{[0,1]}\{[1 - F_{u,\alpha}(s_{B}(u,\alpha))]^{p+1} \\ 			+ [F_{u,\alpha}(s_{B}(u,\alpha))]^{p+1} -[1 -
				F_{u,\alpha}(x_{u,\alpha}^{\lambda})]^{p+1} - [F_{u,\alpha}(x_{u,\alpha}^{\lambda})]^{p+1} \}\dif\nu(\alpha) \dif\mathcal{V}_{p}(u). 		\end{aligned} 	\end{equation} 	 Let us consider the function $f :
		[0,1]\rightarrow\mathbb{R}$ with $f(x) = (1-x)^{p+1} + x^{p+1}$. 	Now if $s_{B}(u,\alpha)\leq s_{A}(u,\alpha)$, we have $s_{B}(u,\alpha)\leq x_{u,\alpha}^{\lambda}$ and $$F_{u,\alpha}(s_{B}(u,\alpha))\leq
		F_{u,\alpha}(x_{u,\alpha}^{\lambda})\leq 1/2.$$ Considering $f$ as in \eqref{f}, since it is decreasing in $[0,1/2]$  we have  $f(F_{u,\alpha}(s_{B}(u,\alpha)))\geq f(F_{u,\alpha}(x_{u,\alpha}^{\lambda}))$. That implies that
		the integrand in \eqref{ecuacionSimplicial} is non-negative. The same conclusion is reached in the case $s_{B}(u,\alpha)\geq s_{A}(u,\alpha)$, using the fact that $f$ is increasing in $[1/2,1]$. 	Thus $$D_{mS}((1-\lambda)A +
		\lambda B;\mathcal{X})- D_{mS}(B;\mathcal{X})\ge 0.$$ 	\vspace{.4cm}
		
		\emph{Property P3a for $D_{FS}.$} Let $B\in\mathcal{F}_{c}(\mathbb{R}^{p})$ and $\lambda\in [0,1]$. By hypothesis, $\mathcal{X}\in C^{0}[\mathcal{F}_{c}(\mathbb{R}^{p})]$ is   $F$-symmetric, with
		respect to $A$. 	
		Using  \eqref{ecuacionMSDFp} and
		$x_{u,\alpha}^\lambda$ as in \eqref{x}, we have that 	\begin{equation} 		\begin{aligned}\nonumber 			&D_{FS}((1-\lambda)\cdot A + \lambda\cdot B) - D_{FS}(B;\mathcal{X}) = \\ \nonumber 			
				&\inf_{u\in\mathbb{S}^{p-1}}\int_{[0,1]}1 - (1 - F_{u,\alpha}(x_{u,\alpha}^{\lambda}))^{p+1} - F_{u,\alpha}(x_{u,\alpha}^{\lambda})^{p+1} \dif\nu(\alpha) -\\\nonumber 			&\inf_{u\in\mathbb{S}^{p-1}}\int_{[0,1]}1 - (1 -
				F_{u,\alpha}(s_{B}(u,\alpha)))^{p+1} - F_{u,\alpha}(s_{B}(u,\alpha))^{p+1} \dif\nu(\alpha). 		\end{aligned} 	\end{equation} Following the arguments in the proof of Property P3a for $D_{mS}$, 	\begin{equation} 	
			\begin{aligned}\nonumber 		&\int_{[0,1]}1 - (1 - F_{u,\alpha}(x_{u,\alpha}^{\lambda}))^{p+1} - F_{u,\alpha}(x_{u,\alpha}^{\lambda})^{p+1} \dif\nu(\alpha)\geq\\\nonumber 		&\int_{[0,1]}1 - (1 -
				F_{u,\alpha}(s_{B}(u,\alpha)))^{p+1} - F_{u,\alpha}(s_{B}(u,\alpha))^{p+1} \dif\nu(\alpha) 	\end{aligned}
		\end{equation}
		for each $u\in\mathbb{S}^{p-1}$. The inequality is preserved if we take the infimum on both sides. Thus $D_{FS}((1-\lambda)\cdot A + \lambda\cdot B; \mathcal{X})\geq D_{FS}(B;\mathcal{X})$. 	\vspace{.4cm}
		
		\emph{Property P3b for $D_{mS}$ and   $D_{FS}.$} In \cite[Theorem 5.4]{primerarticulo}, it is proved that P3b is equivalent to P3a for any $\rho_{r}$ metric with $r\in (1,\infty)$.
	\end{proof}

	\begin{proof}[Proof of Proposition \ref{proposition1BMSDp}]
		Let $A,B\in\mathcal{F}_{c}(\mathbb{R}^{p})$ be two fuzzy sets such that $A$ maximizes $D_{FS}(\cdot ;\mathcal{X})$. 	 Any $C_u$ defined 	 as appears in the definition of $\mathfrak{B}$ satisfies 	
		$C_u\subseteq[0,1]$ and $\nu(C_u)=1.$ Thus,
		\begin{equation*}
			D_{FS}(A+n\cdot B;\mathcal{X}) = \inf_{u\in\mathbb{S}^{p-1}} \int_{C_{u}} \mathbb{P}(s_{A+n\cdot B}(u,\alpha)\in [m_{\mathcal{X}}(u,\alpha), M_{\mathcal{X}}(u,\alpha)]) \dif\nu(\alpha)
		\end{equation*}
		and, fixing an arbitrary  $u\in\mathbb{S}^{p-1}$,
		\begin{equation*}
			D_{FS}(A+n\cdot B;\mathcal{X})\leq \int_{C_{u}} \mathbb{P}(s_{A+n\cdot B}(u,\alpha)\in [m_{\mathcal{X}}(u,\alpha), M_{\mathcal{X}}(u,\alpha)]) \dif\nu(\alpha).
		\end{equation*}
		Using the Dominated Convergence Theorem, we obtain
		\begin{equation}\label{ecuacion2MSD4}
			\lim_{n\rightarrow\infty} 	D_{FS}(A+n\cdot B;\mathcal{X})\leq\int_{C_{u}}\lim_{n\rightarrow\infty}   \mathbb{P}(s_{A+n\cdot B}(u,\alpha)\in [m_{\mathcal{X}}(u,\alpha), M_{\mathcal{X}}(u,\alpha)])\dif\nu(\alpha).
		\end{equation}
		Making use of Proposition \ref{Tecuacion2MSDp} and  \eqref{soportesuma},
		\begin{equation}\label{inT}
			\begin{aligned} 		&\mathbb{P}(s_{A+n\cdot B}(u,\alpha)\in [m_{\mathcal{X}}(u,\alpha), M_{\mathcal{X}}(u,\alpha)]) =  \\ 		&1 - [1 - F_{u,\alpha}(s_{A}(u,\alpha) + n\cdot s_{B}(u,\alpha))]^{p+1} -
				[F_{u,\alpha}(s_{A}(u,\alpha) + n\cdot s_{B}(u,\alpha))]^{p+1}. 	\end{aligned}
		\end{equation}
		As $F_{u,\alpha}$ is the distribution function of the real random variable $s_{\mathcal{X}}(u,\alpha),$ we get, for each $\alpha\in C_{u},$ that the $
		\lim_{n\rightarrow\infty} F_{u,\alpha}(s_{A}(u,\alpha) + n\cdot s_{B}(u,\alpha))
		$ is 1 if $s_{B}(u,\alpha) > 0$ and 0
		if $s_{B}(u,\alpha) < 0.$ Since $B\in\mathfrak{B},$ we have  $s_{B}(u,\alpha)\neq 0$ for all $\alpha\in C_{u}$.
		Making use of this in \eqref{inT}, whether $s_B(u,\alpha)$ is larger or smaller than $0$ we get $$
		\lim_{n\rightarrow\infty} \mathbb{P}(s_{A+n\cdot B}(u,\alpha)\in [m_{\mathcal{X}}(u,\alpha), M_{\mathcal{X}}(u,\alpha)]) = 0,
		$$ for every $\alpha\in C_{u},$ which implies, by \eqref{ecuacion2MSD4}, that $\lim_{n} D_{FS}(A+n\cdot B;\mathcal{X}) = 0$.
		
		The proof for $D_{mS}$ is analogous.
	\end{proof}

	\begin{proof}[Proof of Theorem \ref{SimplicialSDp}]
		\ \\
		\emph{Property P1.} The proof is  analogous to that of P1 in Theorem \ref{teoremamSDp}.
		\vspace{.4cm}
		
		\emph{Property P4b.} Let $\mathfrak{d} := \{d_{r} : r\in [1,\infty]\}\cup \{\rho_{r}: r\in [1,\infty)\}$ be the set of fuzzy metrics of type $d_{r}$ and $\rho_{r}$. Let us fix $d\in\mathfrak{d}.$ Denoting by $A$ a fuzzy
		set that maximizes $D_{nS}(\cdot;\mathcal{X}),$  let $\{A_{n}\}_{n}$ be a sequence of fuzzy sets such that $\lim_{n} d(A,A_{n}) = \infty.$ As $d\in\mathfrak{d},$ this implies, see \cite[Proposition 8.3.]{primerarticulo},
		that there exists  $u_{0}\in\mathbb{S}^{p-1}$ and $\alpha_{0}\in [0,1]$ such that 	\begin{equation}\label{sa} 	\lim_{n} |s_{A_{n}}(u_{0},\alpha_{0})| = \infty. 	\end{equation} 	
		By \eqref{expressionSDF1}, 	\begin{equation*}		
			D_{nS}(A_{n};\mathcal{X}) \leq 			\mathbb{P}(s_{A_{n}}(u_{0},\alpha_{0})\in
			[m_{\mathcal{X}}(u_{0},\alpha_{0}), M_{\mathcal{X}}(u_{0},\alpha_{0})]), 	\end{equation*} 	which, by Proposition \ref{Tecuacion2MSDp}, results in 	\begin{equation*} 	\begin{aligned} 	D_{nS}(A_{n};\mathcal{X}) \leq & 	
				1 - [1 - F_{u_{0},\alpha_{0}}(s_{A_{n}}(u_{0},\alpha_{0}))]^{p+1} \\ &- [F_{u_{0},\alpha_{0}}(s_{A_{n}}(u_{0},\alpha_{0})) - \mathbb{P}(s_{\mathcal{X}_{1}}(u_{0},\alpha_{0}) = 			s_{A_{n}}(u_{0},\alpha_{0}))]^{p+1}. 		 
		\end{aligned} 	\end{equation*} 	Taking limits in this expression, and making use of  \eqref{sa} and the properties of the cumulative distribution function, we obtain 	 $\lim_{n} D_{nS}(A_{n};\mathcal{X}) = 0$. 	
		\vspace{.4cm}
		
		\emph{Property P4a.} According to \cite[Proposition 5.8]{primerarticulo}, P4b implies P4a for the metrics  $d_r$ and $\rho_r$ for any $r\in[1,\infty).$
	\end{proof}
	
	{\bf Acknowledgments} A. Nieto-Reyes and L.Gonzalez were supported by grant MTM2017-86061-C2-2-P funded by  MCIN/AEI/ 10.13039/501100011033 and ``ERDF A way of making Europe". P. Ter\'an is supported by the Ministerio de Ciencia, Innovaci\'on y Universidades grant PID2019-104486GB-I00 and the Consejer\'\i a de Empleo, Industria y Turismo del Principado de Asturias grant
	GRUPIN-IDI2018-000132.

\end{document}